%% file: main.tex
\documentclass[letterpaper,11pt]{article}
\usepackage[utf8]{inputenc}

\usepackage{amsmath, amsthm, amsfonts, amssymb, thm-restate}
\usepackage{thmtools}
\usepackage{algorithmicx}
\usepackage[table,xcdraw]{xcolor}
\usepackage[ruled,vlined,linesnumbered]{algorithm2e}
\usepackage{setspace}
\usepackage{mathtools}
\usepackage[numbers]{natbib}
\usepackage{comment} 
\usepackage[most]{tcolorbox} 
\usepackage{xfrac}
\usepackage{hyperref}
\usepackage{multirow}
\usepackage{caption}
\usepackage{bm}
\usepackage{newfloat}
\usepackage{enumitem}
\usepackage{dblfloatfix} 
\usepackage{wrapfig}
\usepackage{csquotes}

\usepackage{fancyhdr}

\SetCommentSty{mycommfont}
\let\oldnl\nl
\newcommand{\nonl}{\renewcommand{\nl}{\let\nl\oldnl}}

\usepackage[margin=1in]{geometry}
\allowdisplaybreaks

\definecolor{mygreen}{RGB}{10,170,100}
\definecolor{myred}{RGB}{10,170,100}

\hypersetup{
     colorlinks=true,
     citecolor= mygreen,
     linkcolor= myred,
     urlcolor= black
}

\input{macro.tex}
\input{utils.tex}

\makeatletter
\renewcommand{\paragraph}{%
  \@startsection{paragraph}{4}%
  {\z@}{10pt}{-1em}%
  {\normalfont\normalsize\bfseries}%
}
\makeatother

\makeatletter
\patchcmd{\@algocf@start}
  {-1.5em}
  {0pt}
  {}{}
\makeatother

\title{Stochastic Matching via In-n-Out Local Computation Algorithms}

\author{
Amir Azarmehr \and 
Soheil Behnezhad \and
Alma Ghafari \and
Ronitt Rubinfeld
}

\date{}

\begin{document}

\maketitle

\thispagestyle{empty}
\input{abstract}

{
\clearpage
\hypersetup{hidelinks}
\vspace{1cm}
\renewcommand{\baselinestretch}{0.1}
\setcounter{tocdepth}{2}
\tableofcontents{}
\thispagestyle{empty}
\clearpage
}

\setcounter{page}{1}
\input{intro}
\input{preliminaries}

\input{stochastic-matching}
\input{stochastic-matching-lca}

\section*{Acknowledgments}
Part of this work was conducted while the authors were visiting the Simons Institute for the Theory of Computing as participants in the Sublinear Algorithms program.

\bibliographystyle{plainnat}

\bibliography{references}

\clearpage
\appendix
\input{proof-of-f}

\end{document}

%% file: macro.tex
\renewcommand{\epsilon}{\varepsilon}

\DeclareMathOperator{\E}{\ensuremath{\normalfont \textbf{E}}}

\DeclareMathOperator{\opt}{\textsc{Opt}}
\DeclareMathOperator{\al}{\mathcal{A}}
\newcommand{\Xt}{\widetilde{X}}

\newcommand{\smparagraph}[1]{\vspace{0.1cm}\noindent\textbf{#1}}

\newcommand{\hiddencomment}[1]{}

\newcommand{\mc}[1]{\ensuremath{\mathcal{#1}}}

\newcommand{\G}{\mathcal{G}}

\newcommand{\tDelta}{\widetilde{\Delta}}
\DeclarePairedDelimiter\card{\lvert}{\rvert}

\newcommand{\mG}{\mathcal{G}}
\newcommand{\mA}{\mathcal{A}}
\newcommand{\mB}{\mathcal{B}}
\newcommand{\qq}{\bm{q}^-}
\newcommand{\bone}{\mathbf{1}}
\newcommand{\floor}[1]{\left\lfloor #1 \right\rfloor}

\newcommand{\RGMIS}{\ensuremath{\mathsf{RGMIS}}}
\newcommand{\GMIS}{\ensuremath{\mathsf{GMIS}}}
\newcommand{\TMIS}{\ensuremath{\mathsf{TMIS}}}
\newcommand{\isValid}{\ensuremath{\mathsf{isValid}}}
\newcommand{\isInMIS}{\ensuremath{\mathsf{isInMIS}}}
\newcommand{\isInMatching}{\ensuremath{\mathsf{isInMatching}}}

\DeclareMathOperator{\MM}{MM}

%% file: utils.tex
\newcommand{\poly}{\mathrm{poly}}
\newcommand{\quasipoly}{\mathrm{quasipoly}}

\usepackage[noabbrev,nameinlink]{cleveref}
\crefname{lemma}{Lemma}{Lemmas}
\crefname{theorem}{Theorem}{Theorems}
\crefname{result}{Result}{Results}
\crefname{property}{Property}{Properties}
\crefname{claim}{Claim}{Claims}
\crefname{definition}{Definition}{Definitions}
\crefname{observation}{Observation}{Observations}
\crefname{proposition}{Proposition}{Propositions}
\crefname{assumption}{Assumption}{Assumptions}
\crefname{line}{Line}{Lines}
\crefname{figure}{Figure}{Figures}
\creflabelformat{property}{(#1)#2#3}
\crefname{equation}{}{}
\crefname{section}{Section}{Sections}
\crefname{appendix}{Appendix}{Appendices}
\crefname{algCounter}{Algorithm}{Algorithms}
\Crefname{algCounter}{Algorithm}{Algorithms}

\newtheorem{theorem}{Theorem}
\newtheorem{result}{Result}
\newtheorem{lemma}{Lemma}[section]

\newtheorem{definition}[lemma]{Definition}
\newtheorem{claim}[lemma]{Claim}

\newtheorem{remark}[lemma]{Remark}
\newtheorem{assumption}[lemma]{Assumption}

\definecolor{mylightgray}{RGB}{240,240,240}

\algnewcommand{\IIf}[2]{\textbf{if} #1 \textbf{then} #2}
\algnewcommand{\EndIIf}{\unskip\ \algorithmicend\ \algorithmicif}
\algnewcommand{\IElse}[1]{\textbf{else} #1}

\newenvironment{graytbox}{
\par\addvspace{0.1cm}
\begin{tcolorbox}[width=\textwidth,
                  boxsep=5pt,
                  left=1pt,
                  right=1pt,
                  top=2pt,
                  bottom=2pt,
                  boxrule=1pt,
                  arc=0pt,
                  colback=mylightgray,
                  colframe=black,
                  ]
}{
\end{tcolorbox}
}

\newenvironment{whitetbox}{
\par\addvspace{0.1cm}
\begin{tcolorbox}[width=\textwidth,
                  boxsep=5pt,
                  left=1pt,
                  right=1pt,
                  top=2pt,
                  bottom=2pt,
                  boxrule=1pt,
                  arc=0pt,
                  colframe=black,
                  colback=white
                  ]
}{
\end{tcolorbox}
}

\newcounter{algCounter}

%% file: abstract.tex
\begin{abstract}
{\setlength{\parskip}{0.1cm}

\newcommand{\abscite}[1]{{\color{gray}#1}}

Consider the following {\em stochastic matching} problem. We are given a known graph $G=(V, E)$. An unknown subgraph $G_p = (V, E_p)$ is realized where $E_p$ includes every edge of $E$ independently with some probability $p \in (0, 1]$. The goal is to query a sparse subgraph $H$ of $G$, such that the realized edges in $H$ include an approximate maximum matching of $G_p$.

\smallskip\smallskip
This problem has been studied extensively over the last decade due to its numerous applications in kidney exchange, online dating, and online labor markets. For any fixed $\epsilon > 0$, \abscite{[BDH STOC'20]} showed that any graph $G$ has a subgraph $H$ with $\quasipoly(1/p) = (1/p)^{\poly(\log(1/p))}$ maximum degree, achieving a $(1-\epsilon)$-approximation. A major open question is the best approximation achievable with $\poly(1/p)$-degree subgraphs. A long line of work has progressively improved the approximation in the $\poly(1/p)$-degree regime from .5 \abscite{[BDH+ EC'15]} to .501 \abscite{[AKL EC'17]}, .656 \abscite{[BHFR SODA'19]}, .666 \abscite{[AB SOSA'19]}, .731 \abscite{[BBD SODA'22]} (bipartite graphs), and most recently to .68 \abscite{[DS '24]}. 

\smallskip\smallskip
In this work, we show that a $\poly(1/p)$-degree subgraph can obtain a $(1-\epsilon)$-approximation for any desirably small fixed $\epsilon > 0$, achieving the best of both worlds. 

}

\smallskip\smallskip
Beyond its quantitative improvement, a key conceptual contribution of our work is to connect {\em local computation algorithms} (LCAs) to the stochastic matching problem for the first time. 
While prior work on LCAs mainly focuses on their {\em out-queries} (the number of vertices probed to produce the output of a given vertex), our analysis also bounds the {\em in-queries} (the number of vertices that probe a given vertex). We prove that the outputs of LCAs with bounded in- and out-queries ({\em in-n-out LCAs} for short) have limited correlation, a property that our analysis crucially relies on and might find applications beyond stochastic matchings.
\end{abstract}

\clearpage

%% file: intro.tex
\section{Introduction}

We study the {\em stochastic matching} problem. This is a natural {\em graph sparsification} problem that has been studied extensively over the last decade \cite{blumetal-EC15,AKL16,AKL17,YM18,BR18,soda19,sosa19,sagt19,BlumDHPSS-OR20,BehnezhadDH20,BehnezhadD-FOCS20,BehnezhadBD-SODA22,DerakhshanSaneian-ArXiv23} due to its various applications from kidney exchange \cite{BlumDHPSS-OR20} to online dating and labor markets \cite{BR18}.

\smparagraph{Problem Definition:} Given a graph $G=(V, E)$ and a parameter $p \in (0, 1]$, let $G_p=(V, E_p)$ be a random subgraph of $G$ that includes each edge $e \in E$ independently with probability $p$ (we say $e \in E$ is {\em realized} iff $e \in E_p$).\footnote{In this work, we solve a generalization of this problem where edges may have different realization probabilities.} The goal is to select a subgraph $H$ of $G$ without knowing edge realizations such that $(i)$ $H$ is sparse, and $(ii)$ the \underline{realized} edges of $H$ have as large of a matching as the \underline{realized} edges of the whole graph $G$ in expectation. Formally, using $\mu(\cdot)$ to denote maximum matching size, we desire a sparse choice of $H$ that achieves a large {\em approximation ratio} defined as:
    $$
        \E[\mu(H \cap G_p)]/\E[\mu(G_p)].
    $$

Without the sparsity property $(i)$, one can choose $H = G$ and achieve an exact solution (i.e., approximation ratio 1) trivially. Another extreme solution is to take $H$ to be a maximum matching of $G$ which is extremely sparse, but this only achieves a $p$ approximation as only $p$  fraction of the matching is realized in expectation. A simple argument shows that to obtain any constant approximation (independent of $p$), subgraph $H$ needs to have $\Omega(1/p)$ average degree as otherwise most vertices in $H \cap G_p$ will not have any edges and thus remain unmatched. But can we keep the degrees in $H$ independent of the size of $G$, and obtain a good approximation? This is precisely the question studied in the literature of the stochastic matching problem.

\smparagraph{Motivation:} This setting is motivated by applications where every edge of $G$ has a chance of failure and detecting such failures---referred to as {\em querying} the edge---is time-consuming or costly. In such cases, instead of querying every edge in $G$, one can (non-adaptively) query the edges of a much sparser subgraph $H$ and still identify a large matching. In kidney exchange, edges of $G$ correspond to donor/patient pairs that are {\em potentially} eligible for a kidney transplant (some pairs can be ruled out based on available information such as blood types). But to fully verify the eligibility of an edge, one has to perform the costly operation of mixing the bloods of the donor/patient pairs. The stochastic matching problem can be used to minimize such operations while maximizing the number of kidney transplants. We refer interested readers to the paper of \citet{BlumDHPSS-OR20} for an extensive overview of this application of stochastic matching in kidney exchange.

\smparagraph{Prior Work:} The stochastic matching problem is part of a broader class of problems that aim to sparsify graphs while preserving their various properties under random edge failures. Minimum spanning trees and matroids were studied by \citet*{GoemansV-SODA04} in this exact setting more than two decades ago. Many other graph properties have been studied since then including minimum vertex cover \cite{BehnezhadBD-SODA22}, various packing problems \cite{YM18,YM19}, and shortest paths metric problems such as Steiner tree and the traveling salesperson problem (TSP) \cite{Vondrak07}. 

The first work to study matchings in this setting was that of \citet{blumetal-EC15} who proved every graph $G$ admits a subgraph $H$ of maximum degree $\poly(1/p)$, achieving a $(1/2-\epsilon)$ approximation for any fixed $\epsilon > 0$. The approximation was improved in a long and beautiful line of work which led to the discovery of many unexpected connections between stochastic matchings and other areas of TCS. In particular, after a sequence of works \cite{AKL16,AKL17,soda19,AssadiB19}, \citet*{AssadiB19} showed that a $(2/3-\epsilon)$-approximation can be obtained with a graph of maximum degree $\poly(1/p)$, building on connections to dynamic graph algorithms. It was already observed by \citet*{AKL16} (see also \cite{BehnezhadDH20}) that 2/3-approximation is a natural barrier for the problem based on a connection to Ruzsa-Szemer\'edi graphs, a problem in combinatorics which is now known to have applications across various subfields of TCS \cite{GoelKK12,AlonMS12,BehnezhadG-FOCS24,AssadiK-ArXiv,AssadiBKL-STOC23,FischerLNRRS02}. This barrier was finally broken 5 years ago by \citet*{BehnezhadDH20} who showed, building on a connection to distributed algorithms, that any graph $G$ admits a subgraph $H$ of maximum degree  $\quasipoly(1/p) = (1/p)^{\poly\log(1/p)}$ (but still independent of the size of $G$), achieving a $(1-\epsilon)$-approximation for any fixed $\epsilon > 0$. See \Cref{table:prior-work} for an overview of these bounds.

 To summarize, we know that any graph $G$ admits a subgraph $H$ of max-degree $\poly(1/p)$ achieving a $(2/3-\epsilon)$-approximation \cite{AssadiB19,soda19}, and one of maximum degree $\quasipoly(1/p)$ achieving a $(1-\epsilon)$-approximation \cite{BehnezhadDH20} for any fixed $\epsilon > 0$. A major open problem of the area is:
\begin{center}
    {\em Does every graph $G$ admit a $\poly(1/p)$-degree subgraph achieving a $(1-\epsilon)$-approximation?}
\end{center}

More recently, \cite{BehnezhadBD-SODA22,DerakhshanSaneian-ArXiv23} made progress towards this open problem by showing that it is possible to break the 2/3-approximation with a $\poly(1/p)$-degree choice of $H$. The former obtains a .73-approximation provided that the graph is bipartite \cite{BehnezhadBD-SODA22}, and the latter obtains a .68-approximation for general graphs \cite{DerakhshanSaneian-ArXiv23}. Nonetheless, the open problem above, and equivalently, the best approximation achievable with $\poly(1/p)$ per-vertex queries remain wide open.

In this paper, we answer the open question above affirmatively by proving the following.

\begin{graytbox}
    \begin{result}[formalized as \cref{thm:main}]\label{res:main}
        For every fixed $\epsilon > 0$ and every $p \in (0, 1]$, every graph $G$ has a subgraph $H$ of maximum degree $(1/p)^{\exp\left((1/\epsilon) \ln (1/\epsilon)\right)} = \poly(1/p)$ such that
        $$
            \E[\mu(H \cap G_p)] \geq (1-\epsilon) \cdot \E[\mu(G_p)].
        $$
    \end{result}
\end{graytbox}

\begin{table}
    \centering
    \renewcommand{\arraystretch}{1.2} 
    \setlength{\tabcolsep}{10pt} 
    \begin{tabular}{|l|l|c|l|}
    \hline
        \textbf{Reference} & \textbf{Approximation} & \textbf{Degree} & \textbf{Graph} \\ \hline
        \cite{blumetal-EC15,AKL16} & $1/2-\epsilon$ & \multirow{6}{*}{$\poly(1/p)$} & General \\ \cline{1-2} \cline{4-4}
        \cite{AKL17} & $.52$ & & General \\ \cline{1-2} \cline{4-4}
        \cite{soda19} & $.65$ & & General \\ \cline{1-2} \cline{4-4}
        \cite{AssadiB19} & $2/3-\epsilon$ & & General \\ \cline{1-2} \cline{4-4}
        \cite{BehnezhadBD-SODA22} & .73 & & Bipartite \\ \cline{1-2} \cline{4-4}
        \cite{DerakhshanSaneian-ArXiv23} & .68 & & General \\ \hline\hline
        \cite{BehnezhadDH20} & $1-\epsilon$ & $\quasipoly(1/p) = (1/p)^{\poly\log(1/p)}$ & General\\ \hline\hline
        \textbf{This Work} & $1-\epsilon$ & $\poly(1/p)$ & General \\ \hline
     
    \end{tabular}
    \caption{Overview of prior bounds. Here, $\epsilon > 0$ can be any arbitrarily small constant.}
    \label{table:prior-work}
\end{table}

In order to prove \cref{res:main}, we present a novel analysis of an extremely simple and elegant algorithm of the literature introduced first by \citet*{BehnezhadFHR19} for constructing $H$. The algorithm (see \Cref{alg:main}) draws $\poly(1/p)$ independent realizations of the graph and picks an arbitrary maximum matching of each, then takes $H$ to be the union of these matchings. It is clear from the construction that $H$ will have $\poly(1/p)$ maximum degree. It remains to prove that it obtains a $(1-\epsilon)$-approximation. This is what our work focuses on.

Our approximation analysis crucially relies on {\em local computation algorithms} (LCAs) introduced by \citet{AlonRVX12} and \citet{RubinfeldTVX11}, which we adapt to the stochastic matching problem for the first time. An LCA for a graph problem (e.g. graph coloring) does not return the whole output but rather provides an oracle that upon querying a vertex $v$, returns its part of the solution (say the color of $v$). While LCAs are typically used to solve problems on massive graphs, we use them merely for our approximation analysis and in an entirely different way, to bound correlation.

Traditionally, the main measure of complexity for LCAs has been their {\em out-queries}, which is the number of vertices probed to produce the output of a given vertex. One of the main conceptual contributions of our work is to motivate the study of {\em in-queries} for LCAs as well. Informally, the in-query of a vertex $v$ is the number of vertices $u$ that probe $v$ to produce their output (see \cref{sec:lca-def} for formal definitions of in- and out-queries). In \cref{sec:lca-overview}, we provide an overview of how LCAs with bounded in- and out-queries (which we call {\em in-n-out LCAs} for short) have limited correlation in their outputs. This limited correlation is crucial for our analysis and proving \cref{res:main}. We believe that the techniques we develop in studying in-n-out LCAs  (particularly our generic results of \cref{sec:lca-def}) may find applications beyond the stochastic matching problem (see \cref{sec:perspective}).

\section{Our Techniques}\label{sec:overview}

\subsection{Key Tool: Independence via In-n-Out LCAs}\label{sec:lca-overview}

\smparagraph{Background on LCAs:} LCAs are a type of sublinear-time algorithms that are particularly useful for problems where the output is too large to return completely. Our focus in this work will be particularly on LCAs for graph problems. For example, an LCA for graph coloring does not return the color of every vertex at once, but rather returns the color of each vertex $v$ upon query, after probing parts of the input graph (say the local neighborhood of $v$). Importantly, an LCA should produce consistent outputs even if its memory state is erased after each query. This prohibits the LCA from using previous queries to guide its output for future queries, a property that makes such algorithms particularly useful; see \cite{AlonRVX12}.

\smparagraph{LCAs for Independence:} We use LCAs for a completely different purpose than the motivation highlighted above. In particular, \Cref{alg:main} that we employ for constructing subgraph $H$ is not an LCA, but we use LCAs to analyze its approximation ratio. As we soon overview in \cref{sec:analysis-overview}, the key building block of our analysis is to convert an arbitrary distribution $\mc{D}$ over matchings of a random subgraph $G_p$ to another distribution $\mc{D}'$ of matchings of $G_p$, such that $(i)$ almost every vertex has almost the same marginal probability of getting matched in $\mc{D}$ and $\mc{D}'$, and additionally $(ii)$ most vertices are matched almost independently in $\mc{D}'$ (a property that the original distribution $\mc{D}$ may not have). We achieve this by modifying distribution $\mc{D}$ using an LCA that we design for our specific problem. To prove the independence property, we need to bound not just the {\em out-queries} of our LCA which is the standard measure of LCA complexity, but also its {\em in-queries}. Let us define these terms first.

\smparagraph{In-n-Out Queries and Independence:} Traditionally, the main measure of complexity for LCAs has been their {\em out-queries}. Roughly speaking, this is the number of vertices of the graph that a vertex $v$ probes in order to provide its own output (see \cref{sec:lca-def} for the formal definition). This is natural: the out-query of a vertex corresponds to the time needed to produce the output for that vertex and thus it is natural to minimize this quantity. However, our analysis in this work also relies on bounding the {\em in-queries} for LCAs. The in-query of a vertex $v$, informally speaking here (wih the formal definition presented in \cref{sec:lca-def}), is the number of other vertices $u$ in the graph that would query $v$ to provide their own output.

Let us now discuss the connection between in-n-out queries of LCAs and independence. Suppose that each vertex $v$ of the graph has a private random tape $r_v$ that the LCA learns about the moment that it probes $v$, and suppose that this is the only source of randomization available to the LCA. Now take two vertices $u$ and $v$, and let $Q^+(v)$ and $Q^+(u)$ denote the set of vertices that $u$ and $v$ respectively probe (and thus learn their random tapes) in order to produce their own outputs. It is clear that if we always have $Q^+(v) \cap Q^+(u) = \emptyset$, then the outputs of $u$ and $v$ must be independent. Therefore, bounding out-queries (which is the size of set the $Q^+(\cdot)$) is helpful for bounding independence. But note that it is not enough. Consider a hypothetical scenario where $|Q^+(v)|=1$ for every vertex $v$, but all all vertices query the same vertex $w$. This way $Q^+(v) \cap Q^+(u) = \{w\}$ for all pairs $u, v$ and thus we cannot argue any independence for any pair. A useful observation here is that in-queries of $w$, which is the size of the set $Q^-(w)$ of vertices that query $w$, is extremely large. Generally, it can be shown that if $|Q^+(v)| \leq Q^+$ and $|Q^-(v)| \leq Q^-$ for {\em every} vertex $v$, then every vertex $v$ is independent of all but at most $Q^+ \cdot Q^-$ other vertices. This is because the set $\{u \mid \exists w \in Q^+(v) \text{ s.t.} u \in Q^-(w)\}$ has size at most $Q^+ \cdot Q^-$ and includes all vertices $u$ that are not independent of $v$ (i.e., $Q^+(v) \cap Q^+(u) \not= \emptyset$).

We note that two main problems arise when applying this argument to achieve independence via LCAs. First, the set $Q^+(v)$ is usually not deterministic, and the random tape of each vertex that $v$ probes guides its choice of the next vertices to probe. Therefore, the requirement of having $Q^+(v) \cap Q^-(u) = \emptyset$ with probability 1 (needed to argue full independence) is too strong. Second, the independence argument above relies on bounding the in- and out-queries of {\em all} vertices at the same time. This seems hard to achieve, and we can only bound the {\em expected} size of $Q^-(v)$ for every vertex $v$. To resolve the first issue, instead of full independence, we focus on almost independence (which we measure by bounding the total variation distance to full independence). This helps because even if $Q^+(v)$ and $Q^+(u)$ intersect, but with a small probability, we can still argue that the outputs of $u$ and $v$ are almost independent. To resolve the second issue, we prove a general lemma (\cref{lem:correlated_ub}) that bounding in-n-out queries in expectation is sufficient for bounding dependencies of outputs. We emphasize that this result regarding independence guarantees via LCAs is generic (i.e., not specific to our application of stochastic matchings) and may find other applications. For this reason, we present it in \cref{sec:lca-def} as a standalone result.

\subsection{Overview of Our Approximation Analysis}\label{sec:analysis-overview}

Let us now get back to our stochastic matching problem and overview how we use LCAs to analyze the approximation ratio achieved by subgraph $H$. 

Let us first recall the construction of $H$ that we discussed earlier. For a parameter $R$ (which is $\poly(1/p)$ in our case), the algorithm first draws $R$ realizations $\mG_1, \mG_2, \ldots, \mG_R$ of $G$. These realizations are independent of each other and of the actual realization $G_p$.
Then, for each $\mG_i$ a maximum matching $M_i = \MM(\mG_i)$ is computed, where $\MM$ can be any fixed maximum matching algorithm.
Finally, $H$ is taken to be the union of $M_i$ for $1 \leq i \leq R$.
Observe that $H$ has a maximum degree of at most $R$.
The crux of our paper is centered around lower bounding the expected size of the maximum realized matching of $H$, and relating it to that of $G$.

To analyze the approximation ratio, we construct a large fractional matching $y$ in $H_p = H \cap G_p$ (which we show also satisfies the blossom inequalities, hence has no integrality gap).
For an edge $e \in E$, let $q_e$ be the probability that $e \in \MM(G_p)$.
Note that $\card{q} = \sum_e q_e = \E[\mu(G_p)]$ is the optimal value against which the approximation ratio is defined.
We partition the edges into two groups based on their matching probability $q_e$ and a threshold  $\tau < \epsilon p$ that we specify later:
(1) the \emph{crucial} edges $C$ which have a large matching probability $q_e \geq \tau$,
and (2) the \emph{non-crucial} edges $N$ which have a small matching probability $q_e < \tau$.

To construct the fractional matching $y$, we treat the crucial and non-crucial edges completely separately as they have very different properties. The nice property of crucial edges is that so long as $R$ is sufficiently larger than $1/\tau$, almost all crucial edges  appear in $H$ by its construction. On the other hand, the nice property of non-crucial edges is that they have small $q_e$. We heavily rely on this property in constructing the fractional matching on the non-crucial edges, which we essentially show can be used interchangeably. 

First, a fractional matching $f$ is obtained on the non-crucial edges and an integral matching $M_C$ is constructed in $C_p = C \cap G_p$ (the set of realized crucial edges) using an LCA.
Then, the two are combined to obtain a set of values $x$ on the edges of $H_p$ as follows.
For all edges $e \in M_C$, the value of $x_e$ is set to $1$, and for the edges adjacent to $M_C$, the value is set to $0$.
For all the other edges, the value of $x_e$ is set to $f_e$ and scaled up. The resulting set of values $x$ is not a fractional matching but satisfies similar constraints in expectation.
We use a Rounding Lemma (\cref{lem:rounding}) to convert $x$ to the desired fractional matching $y$.

Let us now discuss how the integral matching $M_C$ is constructed and how LCAs help with it. Consider the distribution $\mc{D}$ of matchings on the set $C$ defined as $\MM(G_p) \cap C$. Namely, draw a realization, and take its crucial edges in the maximum matching. The problem with this distribution is that whether a vertex $v$ is matched or not in it, can be vastly correlated with other vertices (recall that $\MM(\cdot)$ is an arbitrary maximum matching algorithm). However, in order to argue that non-crucial edges can combine well with this matching it is necessary \cite{BehnezhadDH20,AKL17} to draw the matching $M_C$ with independence. Therefore, we design an in-n-out LCA that transforms $\mc{D}$ to another distribution $\mc{D}'$ that matches vertices with (almost) the same probability but in addition satisfies independence.
 
The most significant difference between our analysis and that of \cite{BehnezhadDH20} is the construction $M_C$.
In their paper, $M_C$ satisfies three conditions:
(1) $\Pr[v \in V(M_C)] \leq \Pr[v \in V(\MM(G_p) \cap C)]$ for all vertices $v \in V$, 
(2) $\E[\card{M_C}] \geq (1-\epsilon) \E[\card{\MM(G_p) \cap C}]$,
and $(3)$ for all vertices $u, v \in V$ that are further than $\lambda = \Theta_\epsilon(\log \Delta_C)$ from each other in $C$, the events $u \in V(M_C)$ and $v \in V(M_C)$ are independent.
Here, $\Delta_C$ denotes the maximum degree in $C$.
The third property is crucial in proving that not too much of $f$ is lost when it is combined with $M_C$.
They guarantee this property by using a LOCAL algorithm for $M_C$ that runs in $O_\epsilon(\log \Delta_C)$ rounds.
As a result, for each vertex $v$, the event $v \in V(M_C)$ can be correlated with $\Delta_C^{O_\epsilon(\log \Delta_C)}$ others.
This quasi-polynomial dependence on $\Delta_C$, causes the quasi-polynomial dependence of $R$ on $1/p$.

In contrast, we use an LCA to construct $M_C$.
The out-query of the LCA is bounded deterministically by $\Delta_C^{\poly(1/\epsilon)}$.
However, LCAs do not benefit from the same independence guarantees of LOCAL algorithms.
Therefore, we need to quantify the amount of dependence between events of form $v \in V(M_C)$, thereby bounding how much the values of $x$ can deviate from the case where all the events are independent.
In particular, we show for every vertex $v$, that the expected number of in-queries is also at most $\Delta_C^{\poly(1/\epsilon)}$, and as a result the expected size of its correlated set $\Psi(v)$ (intuitively, the set of vertices the output of which is correlated with $v$) is at most only $\Delta_C^{\poly(1/\epsilon)}$. This yields an $R$ that depends only polynomially on $1/p$.

Now, we give an overview of the LCA algorithm for $M_C$. It has a recursive structure. We use $\mB(C_p, r)$ to denote its output matching where $r$ is recursion depth.
For the final output $M_C$, the recursion level is set to a sufficiently large $\poly(1 / \epsilon)$.
In the base case $\mB(C_p, 0)$, the output is simply an empty set. Each level is obtained by augmenting the matching of the previous level.
We defer the details of the recursion to the body of the paper. On a high level, it operates as follows:
\begin{enumerate}
    \item A set of augmenting hyperwalks are defined based on $\mB(C_p, r - 1)$. 
    \item An approximate maximal independent set of the hyperwalks is chosen to be applied to $\mB(C_p, r - 1)$.
    \item Applying the hyperwalks yields $\mB(C_p, r)$.
\end{enumerate}

For an edge $e$, this recursion can be implemented as an LCA as follows: (1) the algorithm checks whether $e \in \mB(C_p, r-1)$, and (2) the algorithm goes over all the hyperwalks containing $e$ and checks whether they are applied to $\mB(C_p, r-1)$ by invoking an LCA for maximal independent set on a graph where the vertices correspond to the hyperwalks.
Recall that the augmenting hyperwalks are defined based on $\mB(C_p, r-1)$. Therefore, for a hyperwalk $W$, the LCA has to check whether $W$ is a valid hyperwalk on the fly by checking $e' \in \mB(C_p, r-1)$ for all $e' \in W$.
The number of out-queries can be bounded deterministically, by bounding the number of out-queries for the MIS LCA.
The number of in-queries requires a more delicate analysis.
To do so, for each level $r$, we go over all the possible random tapes of the algorithm and bound the number of queries made to each edge on level $r-1$ by a counting argument.

\section{Bounding Correlation via In-n-Out LCAs}

In \cref{sec:lca-def}, we formally define LCAs, their in- and out-queries, and a notion of {\em correlated set} that we later show relates to in- and out-queries. In \cref{sec:imp-c-set} we show how bounding the correlated set results in almost independence of the output of most vertices. Finally, in \cref{sec:bounding-c-set}, we present a generic method of translating (expected) upper bounds on in- and out-queries to bounds on the correlated set. 

We hope that the tools developed in this section find applications beyond the stochastic matching problem. As such, we have kept all the definitions and statements of this section independent of stochastic matchings.

\subsection{Definitions}\label{sec:lca-def}

In this section, we start by giving a definition of local computation algorithms (LCAs) for graphs introduced in \cite{RubinfeldTVX11,AlonRVX12}. We start by discussing deterministic LCAs, then focus on randomized ones.

Let $G=(V, E)$ be an input graph. We say $\mc{A}$ is a {\em local computation algorithm} (LCA) that computes a property $f$ of the graph (e.g.\ a coloring), if upon querying a vertex $v$, $\mc{A}$ returns the output $f(v)$ of vertex $v$ (e.g.\ the color of $v$). In this work, we assume that the LCA $\mc{A}$ has adjacency list access to the graph. Namely, the LCA can specify any vertex $u$ and any integer $i$, and receives the $i$-th neighbor of vertex $u$ sorted arbitrarily in a list, provided that $u$ has at least $i$ neighbors. We define the following two complexity measures for any vertex $v$:

\begin{graytbox}
\textbf{Out-Queries} $Q^+(v)$:  The set of vertices for which at least one edge is discovered when computing the output of $v$.\\[0.1cm]

\textbf{In-Queries} $Q^-(v)$: The set of vertices $u$ such that $v \in Q^+(u)$.
\end{graytbox}

Until now, we have only discussed deterministic LCAs. For randomized LCAs, we assume that every vertex $v \in V$ has its own {\em private} tape of randomness $r(v)$. We assume that the LCA learns $r(v)$ the moment that it discovers an edge $(v, u)$ (this edge may have been discovered either by querying $u$ or $v$). We note there are some LCAs in the literature that benefit from public randomness, but our focus in this paper is only on LCAs that rely on private randomness.

It has to be noted that for randomized LCAs, sets $Q^+(v)$ and $Q^-(v)$ are not deterministic, and can be affected by the random tapes of the vertices. In particular, after discovering a vertex $u$, the LCA may choose to query different vertices depending on the random tape $r(u)$.

As discussed earlier, our main motivation for studying in-queries in addition to out-queries is to bound correlations among the outputs. Towards this, we define the following notion of the ``correlated set'' of a vertex and then relate it to independence.

\begin{graytbox}
\begin{itemize}[leftmargin=15pt]
    \item \textbf{Correlated Set} $\Psi(v)$:
    
    This is the set of vertices $u$ such that $Q^+(u) \cap Q^+(v)$ is non-empty.
\end{itemize}
\end{graytbox}

Similar to sets $Q^+(v)$, $Q^-(v)$, the set $\Psi(v)$ is also random for randomized LCAs, and may change after changing the random tapes of the vertices.

For brevity, we use $q^+(v), q^-(v)$, and $\psi(v)$ to respectively denote $|Q^+(v)|, |Q^-(v)|$, and $|\Psi(v)|$.

\newcommand{\cov}[0]{\textsc{Cov}}
\newcommand{\var}[0]{\textsc{Var}}

\subsection{Consequences of Bounding the Correlated Set}\label{sec:imp-c-set}\label{sec:bounding-c-set}

\begin{definition}
    Given probability measures $P$ and $Q$ on the same space $\Omega$, their total variation distance is
    $$
    d_{TV}(P, Q) = \frac{1}{2} \sum_{x \in \Omega} \card{P(x) - Q(x)} 
    = \max_{A \subseteq \Omega} P(A) - Q(A)
    = \min_{\gamma \in \Pi(P, Q)} \Pr_{(X, Y) \sim \gamma}[X \neq Y],
    $$
    where $\gamma$ corresponds to a coupling of $P$ and $Q$.
    We also use $d_{TV}(X, Y)$ for two random variables $X$ and $Y$ to denote the total variation distance between their distributions.
\end{definition}

\begin{claim} \label{clm:difference-by-dtv}
    Take two random variables $X, Y \in \Omega$ and any function $f : \Omega \to [0, M]$. Then, it holds
    $$
    \card{\E[f(X)] - \E[f(Y)]} \leq M \cdot d_{TV}(X, Y).
    $$
\end{claim}
\begin{proof}
    Take a coupling $\gamma$ of $X$ and $Y$ that minimizes the probability $\Pr_{(X, Y) \sim \gamma}[X \neq Y]$, and let $A$ be the event that $X$ and $Y$ are not equal.
    Note that by definition $\Pr[A] = d_{TV}(X, Y)$.
    We have
    \begin{align*}
    \card{\E[f(X)] - \E[f(Y)]}
    &= \card{\E_{(X, Y) \sim \gamma}[f(X) - f(Y)]} \\
    &\leq \Pr[A] \cdot \card{\E_{(X, Y) \sim \gamma}[f(X) - f(Y) \mid A]} + \Pr[\bar{A}] \cdot \card{\E_{(X, Y) \sim \gamma}[f(X) - f(Y) \mid \bar{A}]}\\
    &\leq d_{TV}(X, Y) \cdot M,
    \end{align*}
    where the last inequality follows from the fact that $\card{f(x) - f(y)} \leq M$ for any $x, y \in \Omega$, and that $f(X) = f(Y)$ when $A$ has not happened. This concludes the proof.
\end{proof}

\begin{definition}[$\epsilon$-close to independence]
    Let $(X_1, X_2, \ldots, X_k)$ be a sequence of random variables, and $(\Xt_1, \Xt_2, \ldots, \Xt_k)$ be independent random variables such that $\Xt_i$ has the same marginal distribution as $X_i$ for all $1 \leq i \leq k$.
    Then, $(X_1, X_2, \ldots, X_k)$ are $\epsilon$-close to being independent if 
    $$
    d_{TV}\big((X_1, X_2, \ldots, X_k), (\Xt_1, \Xt_2, \ldots, \Xt_k)\big) \leq \epsilon.
    $$
    Similarly, events $A_1, \ldots, A_k$ are $\epsilon$-close to being independent if their indicator variables $\bone_{A_1}, \ldots, \bone_{A_k}$ are so.
\end{definition}

We now prove claims tailored to our needs for almost independent variables by utilizing \cref{clm:difference-by-dtv}.

\begin{claim} \label{clm:conditional-probability}
Let $A$ and $B$ be two events $\epsilon$-close to being independent.
Then, it holds:
$$
\Pr[A \mid B] \leq \Pr[A] + \frac{\epsilon}{\Pr[B]}.
$$
\end{claim}
\begin{proof}
    It holds:
    \begin{align}
    \Pr[A \mid B]
    &= \frac{\Pr[AB]}{\Pr[B]} \label{eq:cp-by-def} \\
    &\leq \frac{\Pr[A]\Pr[B] + \epsilon}{\Pr[B]} \label{eq:cp-close-to-independent} \\
    &= \Pr[A] + \frac{\epsilon}{\Pr[B]}. \notag
    \end{align}
    Here, \eqref{eq:cp-by-def} is by definition,
    and $\eqref{eq:cp-close-to-independent}$ follows from applying \cref{clm:difference-by-dtv} to $f(\bone_A, \bone_B) = \bone_A\bone_B$.
    More specifically, let $(\widetilde{\bone}_A, \widetilde{\bone}_B)$ be independent random variables with the same marginal distribution as $(\bone_A, \bone_B)$.
    Observe that $d_{TV}\big((\widetilde{\bone}_A, \widetilde{\bone}_B), (\bone_A, \bone_B)\big) < \epsilon$.
    Therefore, we have:
    \begin{align*}
        \Pr[AB] &= \E[\bone_{A}\bone_{B}] \\
        &\leq \E[\widetilde{\bone}_A \widetilde{\bone}_B] + \epsilon\tag{by \cref{clm:difference-by-dtv}} \\
        &= \E[\widetilde{\bone}_A] \E[\widetilde{\bone}_B] + \epsilon\tag{by the independence of $\widetilde{\bone}_A$ and $\widetilde{\bone}_B$} \\
        &= \Pr[A] \Pr[B] + \epsilon.
    \end{align*}
    This concludes the proof.
\end{proof}

\begin{claim}\label{clm:cov_ub}
    Let $(X, Y, Z)$ be three variables $\epsilon$-close to independence,
    such that $X \in [0, M_X]$, $Y \in [0, M_Y]$, and $Z \in \{0, 1\}$. Then, it holds:
    $$
    \cov(X, Y \mid Z = 0) \leq \frac{3\epsilon}{\Pr[Z = 0]}  M_X M_Y.
    $$

\end{claim}

\begin{proof}
    Without loss of generality, we can assume $M_X = M_Y = 1$ by rescaling.
    It holds by definition:
    \begin{align*}
        \cov(X, Y \mid Z = 0) 
        &= \E[XY \mid Z = 0] - \E[X \mid Z = 0]\E[Y \mid Z = 0] \\
        &= \frac{\E[XY \bone_{Z = 0}] }{\Pr[Z = 0]} - \frac{\E[X \bone_{Z = 0}]}{\Pr[Z = 0]} \cdot \frac{\E[Y \bone_{Z = 0}]}{\Pr[Z = 0]}
    \end{align*}
    Observing that $(X, Y, \bone_{Z = 0})$ are also $\epsilon$-close to being independent, we can invoke \cref{clm:difference-by-dtv} for $f(X, Y, Z) = XY\bone_{Z = 0}$, $f(X, Z) = X\bone_{Z = 0}$, or $f(Y, Z) = Y\bone_{Z = 0}$ to get:
    \begin{align*}
        \cov(X, Y \mid Z = 0) 
        &\leq \frac{\E[X]\E[Y]\Pr[Z = 0] + \epsilon}{\Pr[Z = 0]}
        - \frac{\E[X]\Pr[Z = 0] - \epsilon}{\Pr[Z = 0]}
        \cdot \frac{\E[Y]\Pr[Z = 0] - \epsilon}{\Pr[Z = 0]} \\
        &= \frac{\epsilon}{\Pr[Z = 0]} + \frac{\epsilon E[Y]}{\Pr[Z = 0]} + \frac{\epsilon E[X]}{\Pr[Z = 0]}
        - \frac{\epsilon^2}{\Pr[Z = 0]^2} \\
        &\leq  \frac{3\epsilon}{\Pr[Z = 0]}. \qedhere
    \end{align*}
\end{proof}

\begin{definition}
    For any two vertices $u, v \in V$, $\delta(u, v) = \Pr[Q^+(u) \cap Q^+(v) \neq \emptyset]$.
\end{definition}

The following lemma is key in connecting the size of the correlated sets to the dependence of the LCA on different vertices.

\begin{lemma}
    \label{lem:distance-from-independence}
    Take an LCA $\mA$, a graph $G = (V, E)$, and vertices $u_1, \ldots, u_k$.
    Let $X_{u_i} \in \{0, 1\}$ be any function of the part of the graph and random tapes the $\mA(G, u_i)$,
    and let $\delta = \sum_{i \neq j} \delta(u_i, u_j)$.
    Then, $X_{u_1}, \ldots, X_{u_k}$ are $\delta$-close to being mutually independent.
\end{lemma}
\begin{proof}
    Let $(\Xt_{u_1}, \ldots, \Xt_{u_k})$ be independent variables with the same marginal distributions as the original variables $(X_{u_1}, \ldots, X_{u_k})$.
    To prove the lemma, we present a coupling of $(X_{u_1}, \ldots, X_{u_k})$ and $(\Xt_{u_1}, \ldots, \Xt_{u_k})$ where the independent variables are equal to the original ones whenever 
    $Q^+(u_i) \cap Q^+(u_j) \neq \emptyset$ for all $i \neq j$. Here, $Q^+(u_i)$ refers to the out-queries made to by $\mA(G, u_i)$ to compute $X_{u_i}$.
    The coupling is constructed by two parallel runs of the algorithm.
    One to produce the original variables which runs the LCA as is,
    and one to produce the independent variables which assures independence by using a separate set of random tapes for each $\Xt_{u_i}$ (referred to as the altered LCA).
    
    The coupling proceeds as follows: $X_{u_1}, \ldots, X_{u_k}$ are created one by one by running the original LCA $\mA(G, u_i)$. When computing $X_{u_i}$, if the LCA explores a vertex $v$ already explored during the computation of $X_{u_j}$ for some $j < i$, then $X_{u_i}$ uses the already revealed random tape (i.e.\ the random tapes are consistent between the calls to the LCA).
    Similarly $\Xt_{u_1}, \ldots, \Xt_{u_k}$ are created in parallel by running an altered version of the LCA.
    While there have been no intersections in the out-queries, we let the altered LCA take the exact same actions as the original.
    However, as soon as an already explored vertex $v$ is explored for the second time the altered LCA draws a new independent random tape for $v$ to ensure independence.
    From there, the original and altered LCA diverge and run independently. The original one continues the normal flow of the LCA, and the altered one draws new independent random tapes for every $\Xt_{u_i}$ that follows.

    Observe that the altered and the original LCAs diverge only when there is an intersection in the out-queries. Therefore
    \begin{align*}
    d_{TV}\big((X_1, X_2, \ldots, X_k), (\Xt_1, \Xt_2, \ldots, \Xt_k)\big) 
    &\leq \Pr[\exists i \neq j: Q^+(u_i) \cap Q^+(u_j) \neq \emptyset] \\
    &\leq \sum_{i \neq j} \Pr[Q^+(u_i) \cap Q^+(u_j) \neq \emptyset],
    \end{align*}
    where the second inequality follows from the union bound. This concludes the proof.
\end{proof}

\subsection{Bounding The Correlated Set}

It is easy to prove that if $q^+(v) \leq Q^+$ and $q^-(v) \leq Q^-$ for all choices of $v$ at the same time, then $\psi(v) \leq Q^+ \cdot Q^-$. This is because every vertex in $\Psi(v)$ must belong to the in-query set $Q^-(w)$ of a vertex $w$ in the out-query set $Q^+(v)$ of $v$ by definition. Unfortunately, however, providing such for-all upper bounds seems hard to achieve. Particularly, we do not know how to achieve this for our specific application. 
To get around this, we prove a general lemma asserting that expected upper bounds on $q^+(v)$ and $q^-(v)$ also suffice to upper-bound the correlated set sizes $\psi(v)$.

Our result of this section now reads as follows.

\begin{lemma}\label{lem:correlated_ub}
    Let $\al$ be an LCA running on graph $G= (V,E)$ with maximum degree $\Delta$. Suppose that for positive integers $Q^+$ and $Q^-$, it holds for every vertex $u \in V$ that 
    $$
    \E[{q^+(u)}]\leq Q^+ \qquad \text{and} \qquad \E[q^-(u)] \leq Q^-,$$
    then for every vertex $v$
    $$
    \E[\psi(v)] \leq Q^+ \cdot Q^-.
    $$    
\end{lemma}
\begin{proof}
    For any vertex-subset $A \subseteq V$ and any vertex $w \in A$, we define $X_{A, w}$ as the set of vertices $u \not\in A$ such that $Q^+(u) \cap A \not= \emptyset$ and $w$ is the first vertex in $A$ that is discovered by $\al$ when providing the output of $u$. 
    
      Observe that for any vertex $w \in A$, $w$ itself and each vertex in set $ X_{A,w}$ appears in $Q^-(w)$. However, the opposite is not always true, since a vertex in $Q^-(w)$ might query $w$ after querying some other vertices in $A$. Since $w \notin X_{A,w}$ we have $q^-(w) \geq  |X_{A,w}| +1$.

    Take a vertex-subset $A \subseteq V$ and let $R(A)$ be any condition on the random tapes of vertices in $A$. We prove that for any $w \in A$ we have $\E[|X_{A, w}|+1 \mid R(A)] \leq \E[q^-(w)] $. Since we are conditioning on $R(A)$ it is not a trivial statement.  To see this, take a vertex $u \not\in A$ and let us reveal $Q^+(u)$ vertex by vertex. Note that before this set includes any vertex of $A$, all the random tapes are independent of $R(A)$. Moreover, if the first discovered vertex that intersects with $A$ is not $w$, then $u \not\in X_{A, w}$ by definition of $X_{A, w}$. So  for $u$ to join $X_{A, w}$, a neighbor $w'$ of $w$ discovered by $u$ should discover the edge $(w, w')$. Since $w'$ is not aware of the random tape of $w$ before discovering it, the condition on the random tape of $w$ does not change the probability of $w'$ discovering $w$. From this, we get that
    \begin{align*}
        \E[|X_{A, w}|+1 \mid R(A)] & = 1 + \sum_{u \not\in A} \Pr[u \in X_{A, w} \tag{By definition of $X_{A, w}$ and linearity of expectation.}
        \mid R(A)] \\ &\leq  1 + \sum_{u \not\in A} \Pr[u \in Q^-(w)] \\ &\leq \E[q^-(w)]. \tag{Since $w \in Q^-(w)$.}
    \end{align*} 

    Let us fix a vertex $v$. By definition of $\Psi(v)$ and conditioning on $R(Q^+(v))$ we have:
    \begin{flalign*}
    \E[\psi(v) \mid R(Q^+(v))] & = \sum_{w \in Q^+_v} \E[|X_{Q^+(v),w}|+1 \mid R(Q^+(v))] \\ & \leq \sum_{w \in Q^+_v} \E[q^-(w) ] \\ & \leq q^+(v)\cdot Q^-.
    \end{flalign*}
    We conclude the proof by the law of total expectation, we have:
    \begin{equation*}
    \E[\psi[v]] = \E[\E[ \psi(v) \mid R(Q^+(v))]] \leq \E [q^+(v)Q^-] = \E[q^+(v)]\cdot Q^- \leq Q^+ \cdot Q^-.\qedhere
    \end{equation*}
\end{proof}

\subsection{Perspective: Other Potential Applications of In-n-Out LCAs}\label{sec:perspective}

We believe that LCAs with bounded in- as well as out-queries have many desirable properties, and thus deserve to be studied on their own beyond the stochastic matching problem. Our earlier discussion in this section showcases one of the applications of in-n-out LCAs in achieving independence. Here, we briefly discuss two other potential applications of in-n-out LCAs for tolerant testing and dynamic graph algorithms. These settings are out of the scope of this paper and thus we only briefly discuss them, as our only goal is to further motivate the study of in-n-out LCAs.

\smparagraph{Tolerant Testing:} Suppose that we are given a list $x_1, \ldots, x_n$ of distinct integers and want to test if the sequence is monotone (i.e., $x_i < x_j$ for all $i < j$) or far from monotone in sublinear time. A classic result of \cite{ErgunKKRV98} considered the following algorithm: pick a random index $i \in [n]$, assume that the sequence is monotone and perform a binary search for value $x_i$. If the sequence is indeed monotone, the test should clearly end up in location $i$ for any choice of $i$ after $O(\log n)$ inspections. However, if the sequence is $\epsilon$-far from being monotone (meaning that one has to change $\epsilon$ fraction of entries to make it monotone), the analysis of \cite{ErgunKKRV98} shows that there is an $\Omega(\epsilon)$ probability that the test does not end up in the right index $i$. Hence, repeating the process $O(1/\epsilon)$ times, leads to an $O(\epsilon^{-1} \log n)$ time tester for monotonicity. 

The downside of this tester is that it is not {\em tolerant} (\`a la \citet*{ParnasRR06}). Suppose that $x_{n/2}$ holds the smallest value and all other values are correctly sorted. Given that this sequence is monotone except for one index, it is natural to expect the algorithm to return monotone. However, the binary search algorithm fails half of the time, as it moves to the wrong direction and thus returns non-monotone. Observe that the binary search algorithm fails, because the {\em in-query} of element $x_{n/2}$ is high: all queries probe $x_{n/2}$. To fix this and achieve a tolerant tester, \cite{ParnasRR06} presents another algorithm that can be seen as an LCAs with bounded in-queries (to achieve tolerance) and bounded out-queries (to bound testing time). It is conceivable that such a connection is not specific to the monotonicity testing of a list, and may be generalized to tolerant testing of other properties, hence motivating the study of in-n-out LCAs.

\smparagraph{Dynamic Graph Algorithms:} Another advantage of LCAs with bounded in-n-out queries is that their outputs can be efficiently maintained under modifications (e.g. edge or vertex insertions and deletions) to the underlying graph. Suppose for example that we have computed the output of an LCA $\mc{A}$ on a graph $G$, and then an edge $e$ is inserted or deleted from the graph. Note that if $\mc{A}$ for a vertex $v$ does not query any endpoint of $e$, then insertion or deletion of $e$ does not affect its output. Therefore, once $e$ is modified, it suffices to go over the list of in-query sets of its endpoints (which can be maintained explicitly for every vertex), and re-run the LCA on those vertices. Overall, it is not hard to prove that if $Q^-$ and $Q^+$ respectively upper bound the in- and out-queries of all vertices of LCA $\mc{A}$, then the output of the LCA can be maintained under edge insertions and deletions in $O(Q^- \cdot Q^+)$ time per update.

%% file: preliminaries.tex
\section{Preliminaries}
Given a graph $G = (V, E)$, we use $n$ to denote the number of vertices, $\MM(G)$ to denote the edges of a maximum matching in $G$, and $\mu(G)$ to denote the size of the maximum matching.

\begin{definition}[The Stochastic Maximum Matching Problem]
    In the stochastic maximum matching problem,
    given a graph $G$, and a probability parameter $p \in [0, 1]$,
    the goal is to compute a subgraph $H \subseteq G$
    such that
    $$
    \E[\mu(H \cap G_p)] \geq c \cdot \E[\mu(G_p)].
    $$
    Here, $G_p$ is a random subgraph of $G$ that includes each edge with probability $p$, and $c$ denotes the approximation ratio of the algorithm.
\end{definition}

Let $x: E \to \mathbb{R}$ be a set of values on the edges.
$\card{x}$ is used to denote the sum of $x$ on all the edges.
For an edge set $F \subseteq E$, $x_F$ is used to denote the sum of $x$ on the edges of $F$:
$$
x_F \coloneqq \sum_{e \in F} x_e.
$$
For a vertex $v \in V$, $x_v$ is used to denote the sum of $x$ on the edges adjacent to $v$:
$$
x_v \coloneq \sum_{e \ni v}  x_e,
$$
and $x_v^F$ is used to denote the sum of $x$ on the edges that are adjacent to $v$ and in $F$:
$$
x_v^F \coloneq \sum_{\substack{e \ni v \\ e \in F}} x_e.
$$

A \emph{fractional matching} is a set of values on the edges, $f: E \to [0, 1]$, such that for any vertex $v$ it holds $f_v \leq 1$.
If a fractional matching $f$ for $G$ satisfies the \emph{blossom inequalities} (defined below), then $\mu(G) \geq \card{f}$.
The following lemma provides a continuous version of this fact.

\begin{lemma} [folklore] \label{lem:blossom}
    Let $f$ be a fractional matching that satisfies the blossom inequalities for all vertex sets $S$ with $\card{S} \leq \frac{1}{\epsilon}$, i.e.,
    $$
    \sum_{\substack{(u, v) \in E \\ u, v \in S}} f_{(u, v)} \leq \floor{\frac{S}{2}}.
    $$
    Then, $\mu(G) \geq (1 - \epsilon) \card{f}$.
\end{lemma}

%% file: stochastic-matching.tex
\section{Stochastic Matching via LCAs}
\subsection{Algorithm and Definitions}

\begin{theorem} \label{thm:main}
    Let $G_p=(V, E)$ be any graph and let $G$ include every edge $e \in E$ independently with probability $p_e \geq p$. \Cref{alg:main} computes a subgraph $H = (V, E_H)$ of maximum degree $(1/p)^{\exp\left((1/\epsilon) \ln (1/\epsilon)\right)} = (1/p)^{O_\epsilon(1)}$ 
    such that
    $$
        \E[\mu(E_H \cap E_p)] \geq (1-\epsilon) \E[\mu(E_p)].
    $$
    Recall that in the stochastic matching problem, the algorithm has access to $G$ and $p$, but not $G_p$.
\end{theorem}

\begin{algorithm}[H]\caption{The algorithm for constructing subgraph $H$, due to \cite{BehnezhadFHR19}.}\label{alg:main}
          \textbf{Input:} Graph \( G = (V, E) \), probability \( p_e \) for each edge \( e \in E \).

        Let $R \coloneq 1 / 2\tau^- = (1/p)^{O_{\epsilon}(1)}$ (as in \cref{clm:taus}).
       
       \For{$i = 1$ to $R$}{
        
              Sample a realization \( \G_i \) of \( G \), where each edge \( e \in E \) is included in $\mG_i$ independently with probability \( p_e \).
             
              Compute a maximum matching \( \MM(\G_i) \) on the realization $\mG_i$.
        }
        
         \textbf{Return} \( H = \MM(\G_1) \cup \MM(\G_2) \cup \dots \cup \MM(\G_R) \)
         
\end{algorithm}
We use $\opt$ to denote $\E[\mu(E_p)]$ and for a subgraph $H$, we use $H_p$ to denote the subgraph of $H$ that is realized, i.e.\ $E(H_p) = E(H) \cap E_p$.
Throughout this section, we make the following assumption:
\begin{assumption} \label{assumption}
    $\opt \geq \epsilon n$.
\end{assumption}
This assumption is without loss of generality due to a \enquote{bucketing} reduction of \citet*{AssadiKL19} (see Appendix $B$ of \cite{BehnezhadDH20}). Intuitively, if $\opt = o(\epsilon n)$, then one can randomly put vertices into $\Theta(\opt/\epsilon)$ buckets and contract all vertices assigned to the same bucket. This reduces the number of vertices to $O(\opt/\epsilon)$, but keeps the maximum matching size almost as large as $\opt$. We note that for this reduction to be valid, it is crucial that the algorithm works even when the realization probabilities are different. This is because the vertex contractions may result in parallel edges which can be modeled by a simple graph with different edge realization probabilities.

\subsection{Algorithm Analysis}
First, we make some definitions.
For an edge $e$ let $q_e$ be $\Pr[e \in \MM(G_p)]$, the probability of $e$ being matched in $\MM(G_p)$ where the $G_p$ is obtained by sampling each edge $e'$ of $G$ with probability $p_{e'}$.
Thus, it holds $\card{q} = \E[\card{\MM(G_p)}] = \opt$.
Two thresholds $0 < \tau^- < \tau^+ < 1$ are computed based on the graph $G$, the parameter $\epsilon$, and the sampling probabilities $\{p_e\}_{e \in E}$.
An edge $e \in E$ is \emph{crucial} if $q_e \geq \tau^+$,
and \emph{non-crucial} if $q_e \leq \tau^-$ (the edges with $\tau^- < q_e < \tau^+$ are essentially ignored, see \cref{clm:taus}).
We use $C$ and $N$ to denote the set of crucial and non-crucial edges respectively.

Here, we provide a bird's-eye view of the analysis.
To show $\mu(H_p)$ is large, we eventually create a large fractional matching $y$ on the edges of $H_p$ that satisfies the blossom inequalities.
First, a fractional matching $f$ on the non-crucial edges (\cref{lem:f})
is combined with an integral matching $M_C \subseteq C_p$ (\cref{def:MC}) to obtain a set of values $x$ on the edges of $H_p$ (\cref{def:x}).
These values do not form a fractional matching. However, they satisfy similar constraints in expectation (\cref{clm:expectation_ub}), as well as some tail bounds (\cref{clm:x-tail-bound}).
Therefore, we can use the Rounding Lemma (\cref{lem:rounding}) to obtain a fractional matching $y$ from $x$.
The entries of $y$ are small, except on the edges of the matching $M_C$.
As a result, $y$ satisfies the blossom inequalities and $\mu(H_p)$ is at least as large as $\card{y}$.

Before describing the construction of $y$,
we prove some properties of $q$,
and give a characterization of the thresholds $\tau^-$ and $\tau^+$.

\begin{claim}
    For any vertex $v$, it holds $q_v \leq 1$.
\end{claim}
\begin{proof}
    This is true since for an edge $e \ni v$, $q_e$ is the probability of $v$ being matched using $e$ in $\MM(G_p)$.
    Therefore, $q_v = \sum_{e \ni v}q_e$ is the probability of $v$ being matched by any edge, as is at most $1$.
\end{proof}

\begin{claim}
    The maximum degree in $C$ (the set of the crucial edges) is $\Delta_C \leq \frac{1}{\tau^+}$.
\end{claim}
\begin{proof}
    By definition, for any edge crucial edge $e \in C$, it holds $q_e \geq \tau^+$.
    Therefore, for any vertex $v$, since $q_v \leq 1$,
    there can be at most $\frac{1}{\tau^+}$ crucial edges adjacent to $v$.
\end{proof}

Now, we characterize $\tau^-$ and $\tau^+$.
In broad terms, they satisfy two constraints:
$(1)$ $\tau^-$ is sufficiently smaller than $\tau^+$,
and $(2)$ the total mass of $q$ on the edges $e$ with $q_e \in (\tau^-, \tau^+)$ is small.

\begin{claim} \label{clm:taus}
    For any graph $G = (V, E)$ and realization probabilities $\{p_e\}_{e\in E}$, there exists constants $p^{\exp\left((1/\epsilon) \ln (1/\epsilon)\right)} < \tau^- < \tau^+ < (\epsilon p)^2$ such that 
    \begin{enumerate}
        \item $\tau^- < (\tau^+)^{\poly(1/\epsilon)}$ where the power of $\tau^+$ can be taken arbitrarily large, and
        \item $q_C + q_N \geq (1 - \epsilon)\opt$, where $C$ and $N$ are defined w.r.t.\ $\tau^+$ and $\tau^-$.
    \end{enumerate}
\end{claim}
\begin{proof}
    Let $\tau_0 \coloneq (\epsilon p)^2$ and for $1 \leq i \leq \frac{1}{\epsilon}$, let $\tau_i \coloneq \tau_{i-1}^{\poly(1/\epsilon)}$.
    Let $q_i$ be the sum of $q_e$ for edges $e$ with $q_e \in (\tau_i, \tau_{i-1}]$. That is:
    $$
    q_i \coloneq \sum_{\substack{e \in E \\ q_e \in (\tau_i, \tau_{i-1}]}} q_e.
    $$
    Observe that $\sum_{1\leq i \leq 1/\epsilon} q_i \leq \card{q} = \opt$. Therefore, there is an index $i$ for which $q_i \leq \epsilon \opt$. Let $\tau^- = \tau_i$ and $\tau^+ = \tau_{i-1}$.
    Then, for the crucial and non-crucial edges defined w.r.t.\ $\tau^+$ and $\tau^-$, it holds:
    $$
    q_C + q_N = \opt - q_i \geq (1 - \epsilon)\opt.
    $$
    Also, $\tau^- = q_i = ((\epsilon p)^2)^{(\poly(1/\epsilon))^i} \geq p^{\exp\left((1/\epsilon) \ln (1/\epsilon)\right)}$.
\end{proof}

To begin the construction of $y$, we first obtain the fractional matching $f$ on the non-crucial edges, based on $H$.
A similar fractional matching is constructed in \cite{BehnezhadDH20}.
However, there are slight differences in the definition.
Hence, we have included the proof in \cref{sec:ref} for completeness.

\begin{restatable}[\cite{BehnezhadDH20}]{lemma}{lemf}\label{lem:f}
Let $H = \bigcup_i \MM(\mG_i)$ be the output of \Cref{alg:main} on graph $G$. One can define a set of random variables $\{f_e\}_{e \in E}$ based on $H$ such that:
\begin{enumerate}
    \item for every edge $e \notin E(H) \cap N$, $f_e = 0$,
     \item for every vertex $v$, it holds that $\sum_{e \ni v} f_e \leq q^N_v$,
    \item for every edge $e \in E$, $f_e \leq \frac{1}{\sqrt{\epsilon R}}$,
    \item for every non-crucial edge, $\E[f_e] \leq q_e$, and
    \item it holds $\E[\card{f}] \geq q_N - \epsilon \opt$.
\end{enumerate} 
\end{restatable}

Let us now state the lemma used to construct the integral matching $M_C \subseteq C_p$.
On a high level, it proposes that we can compute a matching $M_C$ that has (almost) the same vertex marginals as $\MM(G_p) \cap C$, and at the same time satisfies certain independence guarantees. This independence is phrased in terms of the size of the correlated sets in the LCA algorithm that computes $M_C$.

\begin{lemma} \label{lem:stochastic-matching-lca}
    Let $G=(V, E)$ be a given graph of maximum degree $\Delta$ and let $G_p$ include every edge $e \in E$ independently with probability $p_e$. Let $\mc{A}(H)$ be any possibly randomized algorithm producing a matching of its input graph $H$. 
    There exists an LCA $\mB(H)$ returning a matching of $H$ such that:
    \begin{enumerate}
        \item $\Pr[v \in \mc{B}(G_p)] \leq \Pr[v \in \mc{A}(G_p)]$, for all $v \in V(G)$. \label{item:vertex-expectation}
        \item $\E[|\mc{B}(G_p)|] \geq \E[|\mc{A}(G_p)|] - \epsilon \card{V}$,
        \label{item:matching-expectation}
        \item For any vertex $v \in V$,
        $$\E[\psi(v)] \leq \Delta^{\poly(1/\epsilon)},$$
        where recall $\psi(v)$ is the size of the correlated set of $v$, i.e., the set of vertices that query the same vertices as $v$ in the LCA.
        \label{item:correlated-sets}
    \end{enumerate}
    All the probablistic statements above are taken over both the randomization of the algorithms $\mc{A}$, $\mc{B}$, as well as the randomization in realization $G_p$.
\end{lemma}

\begin{definition} \label{def:MC}
    Let $\mA(C_p) \coloneq \MM(G_p) \cap C$, i.e.\ the set of crucial edges matched by $\MM$ in the realized graph $G_p$.
    Then, let the matching $M_C \subseteq C_p$ be the output of $\mB(C_p)$, where $\mB$ is the LCA obtained by applying \cref{lem:stochastic-matching-lca} to $\mA$ with parameter $\epsilon' = \epsilon^2$.
\end{definition}
\begin{remark}
    The only information $\mB(C_p)$ needs is: (1) the structure of $C$, (2) the realized subset $C_p \subseteq C$, and (3) $q_e$ for every edge $e \in C$ which is the matching probabilities of algorithm $\mA$. Therefore, the output of $M_C$ is independent of the remaining realizations outside of $C$.
\end{remark}

We restate the definition of $\delta(u, v)$.
\begin{definition} 
    For any two vertices $u$ and $v$,
    let $$\delta(u, v) \coloneq \Pr[Q^+(u) \cap Q^+(v) \neq \emptyset],$$ where $Q^+(u)$ and $Q^+(v)$ denote the explored vertices for LCA calls $\mB(C_p, u)$ and $\mB(C_p, v)$,
    and the probability is over the randomness of algorithm $\mB$.
\end{definition}

Based on $f$ and $M_C$, we define values $\{x_e\}_{e \in E}$. Then, we show that (1) $x$ is (almost) a fractional matching in expectation, (2) $\E[\card{x}]$ is almost as large as $\card{q}$, and (3) for any vertex $v$ the probability of $\Pr[x_v > 1 + O(\epsilon)]$ is small.

\begin{definition} \label{def:x}
    Let $M_C$ be the output of $\mB$ on the realized edges of $C$. Then, for every crucial edge $e \in C$
    $$
        x_e \coloneq 
        \begin{cases}
        1, & \text{if } e \in M_C \text{ and } e \in H_p, \\
        0, & \text{otherwise},
        \end{cases}
    $$
    and for every non-crucial edge $e \in N$,
    $$
        x_e \coloneq 
        \begin{cases}
        0, 
        & \text{if $\delta(u, v) > (\epsilon p)^{15}$,}  \\
        
        & \text{or $\Pr[u \notin V(M_C)] < \epsilon^2$ or $\Pr[v \notin V(M_C)] < \epsilon^2$},   \\
        & \text{or $e \notin E_p$ or  $u \in V(M_C)$ or $v \in V(M_C)$, and} \\
        \\
        \frac{f_e}{p_e\Pr[u \notin V(M_C)]\Pr[v \notin V(M_C)]}, & \text{otherwise}.
        \end{cases}
    $$
\end{definition}

\begin{claim}\label{clm:expectation_ub}
Let $x$ be as defined in \cref{def:x}. For any vertex $v \in V$, it holds that $\E[x_v] \leq 1 + \epsilon$. 
\end{claim}

\begin{proof}
Let $M_C$ be the matching defined in \cref{def:MC}. Take any vertex $v \in V$. First, we prove that $\E[x_v \mid v \in V(M_C)]  = 1$. 
If $v \in V(M_C)$, then $x$ is equal to $1$ on the matching edge of $v$ in $M_C$ (when the edge appears in $H_p$), and zero on all the other edges adjacent to $v$.
Therefore,
$$
\E[x_v \mid v \in V(M_C)] \leq 1.
$$

It remains to bound the expectation conditioned on $v \notin V(M_C)$. 
If we have $\Pr[v \notin V(M_C)] < \epsilon^2$, then $x$ is zero on all the adjacent edges of $v$ (by definition of $x$).
Hence, we restrict our attention to the case where $\Pr[v \in V(M_C)] \geq \epsilon^2$.
Let $e_1, \ldots, e_k$ be the non-crucial edges adjacent to $v$,
such that $\Pr[u_i \notin V(M_C)] \geq \epsilon^2$
and $\delta(u_i, v) \leq (\epsilon p)^{15}$,
where $u_i$ is the other endpoint of $e_i$.

Observe that since $v \notin V(M_C)$, the value of $x$ is zero on the crucial edges adjacent to $v$.
Therefore, by linearity of expectation, we get 
$$\E[x_v\mid  v \notin V(M_C)] = \sum_{1 \leq i \leq k} \E[x_{e_i} \mid v \notin V(M_C)]. $$
Given that $v \notin V(M_C)$, the value of $x_{e_i}$ is $\frac{f_{e_i}}{p_{e_i}\Pr[u_i \notin V(M_C)]\Pr[v \notin V(M_C)]}$ if $e_i$ is realized and $u_i \notin V(M_C)$. Otherwise, $x_{e_i}$ is zero. Therefore,
$$
E[x_{e_i} \mid v \notin V(M_C)]  =   \Pr[u_i \notin V(M_C),\: e_i \in E_p \mid  v \notin V(M_C) ] \frac { \E[f_{e_i}]}{p_{e_i} \Pr[u_i \notin V(M_C)]\Pr[v \notin V(M_C)]}.
$$

Recall $M_C$ was computed using $\mB$ based on $C$ and its realization $C_p$.
As a result, the realization of edge $e_i$ is independent from $u_i \notin V(M_C)$ and $v \notin V(M_C)$, and we have

\begin{align*}
E[x_{e_i} \mid v \notin V(M_C)] 
&= \frac {p_{e_i} \Pr[u_i \notin V(M_C) \mid  v \notin V(M_C) ]}{p_{e_i} \Pr[u_i \notin V(M_C)]\Pr[v \notin V(M_C)]} \E[f_{e_i}] \\ 
&= \frac {\Pr[u_i \notin V(M_C) \mid  v \notin V(M_C) ]}{ \Pr[u_i \notin V(M_C)]\Pr[v \notin V(M_C)]} \E[f_{e_i}].
\end{align*} 

Since $\delta(u_i, v) \leq (\epsilon p)^{15}$, by \cref{lem:distance-from-independence}, the events $u_i \notin V(M_C)$ and $ v \notin V(M_C)$ are $(\epsilon p)^{15}$-close to being independent. Hence, by \cref{clm:conditional-probability}, we have
\begin{align*}
E[x_{e_i} \mid v \notin V(M_C)] 
&\leq \frac { \Pr[u_i \notin V(M_C)] +(\epsilon p)^{15}/ \Pr[v \notin V(M_C)]}{\Pr[u_i \notin V(M_C)]\Pr[v \notin V(M_C)]}  \E[f_{e_i}] \\ 
&= \left( 1 + \frac{(\epsilon p)^{15}}{\Pr[u_i \notin V(M_C)]\Pr[v \notin V(M_C)]}  \right) \frac{\E[f_{e_i}]}{\Pr[v \notin V(M_C)]}.
\end{align*} 
By the definition of $e_i$, it holds $\Pr[u_i \notin V(M_C)] \geq \epsilon^2$ and $\Pr[v \notin V(M_C)] \geq \epsilon^2$. Therefore,
$$
E[x_{e_i} \mid v \notin V(M_C)] \leq  (1 + \epsilon)\frac{\E[f_{e_i}]}{\Pr[v \notin V(M_C)]}.
$$
Applying $E[f_{e_i}] \leq q_{e_i}$ from \cref{lem:f}, we have
$$E[x_{e_i} \mid v \notin V(M_C)] \leq  (1 + \epsilon)\frac{q_{e_i}}{\Pr[v \notin V(M_C)]}.$$
Consequently,
$$\E[x_v\mid  v \notin V(M_C)] = \sum_{1 \leq i \leq k} \E[x_{e_i} \mid v \notin V(M_C)] \leq \sum_{1 \leq i \leq k}\frac{(1+\epsilon)q_{e_i}}{\Pr[v \notin V(M_C)]} 
\leq \frac{(1 + \epsilon) q^N_v}{\Pr[v \notin V(M_C)]},
$$
where the last inequality follows from the fact that $e_1, \ldots, e_k$ are non-crucial.
Finally, observe that $\Pr[v \in V(M_C)] \leq q^C_v$ (item \ref{item:vertex-expectation} of \cref{lem:stochastic-matching-lca}), and $q^N_v + q^C_v \leq 1$. Thus, it follows
\begin{equation*}
\E[x_v\mid  v \notin V(M_C)] \leq  \frac{(1 + \epsilon) q^N_v}{1 - q^C_v} \leq (1 + \epsilon). \qedhere
\end{equation*}
\end{proof}

\begin{claim}\label{clm:x-expectation-lb}
    Let $x$ be as defined in \cref{def:x}, it holds that $$\E[|x|] \geq (1-7\epsilon)\opt.$$
\end{claim}

\begin{proof}
    To prove the claim, we show that 
    $$\E[x_C] \geq q_C - 2\epsilon \opt, \qquad \text{and} \qquad \E[x_N] \geq q_N - 4\epsilon \opt.$$
    Then, the claim follows from $q_N + q_C \geq (1-\epsilon)\card{q}$ (\cref{clm:taus}). 
    
    To prove the lower bound on $\E[x_C]$,
    note that $x_C = \card{M_C \cap H}$,
    i.e.\ the number of edges of $M_C$ that \Cref{alg:main} includes in $H$.
    Therefore, we have
    \begin{align}
        \E[x_C]
        &= \E[\card{M_C \cap H}] \notag \\
        &= \sum_{e \in C} \Pr[e \in M_C, e \in H] \notag \\
        &= \sum_{e \in C} \Pr[e \in M_C] \Pr[e \in H], \label{eq:xc-lb}
    \end{align}
    where the last inequality holds because $M_C$ is computed by $\mB(C_p)$ based on the realized edges of $C$, and $H$ is a function of the random tape of \Cref{alg:main}, therefore the two are independent.
    Recall that for $e \in C$ it holds $q_e \geq \tau^+$.
    As a result, $\Pr[e \notin H]$ can be bounded as follows:
    \begin{align*}
        \Pr[e \notin H] 
        &= \prod_{1 \leq i \leq R} \Pr[e \notin \MM(\mG_i)] \\
        &= (1 - q_e)^R \\
        &\leq (1 - \tau^+)^R \\
        &\leq e^{-\tau^+ R} \\
        &\leq \epsilon.
    \end{align*}
    The last inequality follows from $R = \frac{1}{2\tau^-} \geq (1/\tau^+)^{\poly(1/\epsilon)}$.
    Plugging back into \eqref{eq:xc-lb}, we get
    \begin{align}
        \E[x_C] 
        &\geq (1 - \epsilon) \sum_{e \in C} \Pr[e \in M_C] \notag\\
        &= (1 - \epsilon) \E[\card{M_C}] \notag\\
        &\geq (1 - \epsilon) (q_C - \epsilon^2 n) \label{eq:mc-large}\\
        &\geq q_C - 2\epsilon \opt. \label{eq:xc-lb-assumption}
    \end{align}
    Here, \eqref{eq:mc-large} is directly implied by item \ref{item:matching-expectation} of \cref{lem:stochastic-matching-lca},
    and \eqref{eq:xc-lb-assumption} follows from \cref{assumption}.
    
    To prove $\E[x_N] \geq q_N - \epsilon \opt$,
    let $N'$ be the set of non-crucial edges $(u, v)$ such that $\Pr[v \notin V(M_C)] \geq \epsilon^2$, $\Pr[u \notin V(M_C)] \geq \epsilon^2$, and $\delta(u, v) \leq (\epsilon p)^{15}$.
    We first show that $q_{N \setminus N'} \leq 4 \epsilon \opt$, i.e.\ we can ignore the edges not in $N'$.

    Consider an edge $(u, v) \in N \setminus N'$.
    There are two possibilities for $(u, v)$:
    First, $\delta(u, v) > (\epsilon p)^{15}$.
    Let $D$ be the set of all such edges.
    The contribution of these edges to $q$ is at most $q_D \leq \epsilon
    \opt$ by the choice of $\tau^-$ and $\tau^+$ (\cref{clm:small-dependent-edges} below).
    Second, either $\Pr[u \notin V(M_C)] < \epsilon^2$ or $\Pr[v \notin V(M_C)] < \epsilon^2$. Call these edges $F$. Their contribution is at most $\epsilon^2 n$.
    Because for any vertex $u$ with $\Pr[u \notin V(M_C)] < \epsilon^2$, it holds:
    \begin{align}
    \sum_{\substack{e \ni u \\ e \in F}} q_e 
    &\leq \sum_{\substack{e \ni u \\ e \in N}} q_e \label{eq:x-expectation-f-subset-N} \\
    &= q^N_u \label{eq:x-expectation-N-def}\\
    &\leq 1 - q^C_u  \label{eq:x-expectation-qc-plus-qn} \\
    &\leq \Pr[u \notin V(M_C)] \label{eq:x-expectation-from-22} \\
    &\leq \epsilon^2. \notag
    \end{align}
    Here, \eqref{eq:x-expectation-f-subset-N} follows from $F \subseteq N$,
    \eqref{eq:x-expectation-N-def} is by definition,
    \eqref{eq:x-expectation-qc-plus-qn} follows from $q^C_u + q^N_u \leq q_u \leq 1$, and
    \eqref{eq:x-expectation-from-22} holds by item \ref{item:vertex-expectation} of \cref{lem:stochastic-matching-lca}.
    Summing over $u$, we get:
    $$
    q_F \leq \sum_u \sum_{\substack{e \ni u \\ e \in F}} q_e 
    \leq \epsilon^2 n \leq \epsilon \opt.
    $$
    The second ineuqality here follows from $\opt \geq \epsilon n$ (\cref{assumption}).
    Then, we can conclude:
    \begin{equation*}
    q_{N \setminus N'} = q_{D \cup F} \leq q_D + q_f \leq 2\epsilon \opt.
    \end{equation*}

    As the final step, we prove $\E[x_e \mid f_e] \geq (1-\epsilon)f_e$ for all edges $e \in N'$,
    where recall $f$ is the fractional matching that mimics $q$ as defined in \cref{lem:f}.
    Note that for $e = (u, v) \in N'$, the value of $x_e$ is equal to $\frac{f_e}{p_e \Pr[u \notin V(M_C)] \Pr[v \notin V(M_C)]}$ if $e$ is realized and $u, v \notin V(M_C)$, and zero otherwise.
    Conditioning on $f_e$ we get:
    \begin{equation}
    \E[x_e \mid f_e] = \Pr[e \in E_p, u \notin V(M_C), v \notin V(M_C)] \frac{f_e}{p_e \Pr[u \notin V(M_C)] \Pr[v \notin V(M_C)]}.
    \label{eq:x-expectation-conditional-exp}
    \end{equation}
    Recall, that $M_C$ is only determined by the set of realized crucial edges $C_p$ and is independent of non-crucial edges. Hence,
    $$
    \Pr[e \in E_p, u \notin V(M_C), v \notin V(M_C)]
    = p_e \Pr[u \notin V(M_C), v \notin V(M_C)].
    $$
    Also, since $(u, v) \in N'$, the events $u \notin V(M_C)$ and $v \notin V(M_C)$ are $(\epsilon p)^{15}$-close to being independent, and by \cref{clm:difference-by-dtv} it follows:
    $$
    \Pr[e \in E_p, u \notin V(M_C), v \notin V(M_C)]
    \geq p_e (\Pr[u \notin V(M_C)] \Pr[v \notin V(M_C)] - (\epsilon p)^{15}).
    $$
    Plugging this back in \eqref{eq:x-expectation-conditional-exp} implies
    \begin{align*}
    \E[x_e \mid f_e] &\geq p_e (\Pr[u \notin V(M_C)] \Pr[v \notin V(M_C)] - (\epsilon p)^{15}) \frac{f_e}{p_e \Pr[u \notin V(M_C)] \Pr[v \notin V(M_C)]} \\
    &\geq \left(1 - \frac{(\epsilon p)^{15}}{p_e\epsilon^4}\right) f_e \\
    &\geq (1 - \epsilon) f_e.
    \end{align*}

    By lifting the conditioning on $f$, and summing over all the edges, we get:
    \begin{align}
        \E[x_{N'}]
        &= \sum_{e \in N'} \E[\E[x_e \mid f_e]] \notag \\
        &\geq  \sum_{e \in N'} \E[(1-\epsilon)f_e] \notag \\
        &= (1-\epsilon) \E[f_{N'}] \notag \\
        &= (1-\epsilon) \E[f_{N} - f_{N \setminus N'}] \notag \\
        &\geq (1-\epsilon) (q_N - \epsilon \opt - q_{N \setminus N'}) \label{eq:x-expectation-from-f} \\
        &\geq q_N - 4\epsilon \opt, \label{eq:nnp-ub}
    \end{align}
    where \eqref{eq:x-expectation-from-f} follows from 
    $\E[f_N] = \E[|f|] \geq q_N - \epsilon \opt$ and $\E[f_{N\setminus N'}] \leq q_{N \setminus N'}$(\cref{lem:f});
    and \eqref{eq:nnp-ub} utilizes $q_{N\setminus N'} \leq \epsilon \opt$.
    
    Combining this with $\E[x_C] \geq q_C - 2 \epsilon \opt$ and $q_N + q_C > \card{q} - \epsilon\opt$ (\cref{clm:taus}),
    gives:
    \begin{equation*}
    \E[\card{x}] \geq E[x_C] + \E[x_N] \geq (q_C - 2\epsilon\opt) + (q_N - 4\epsilon\opt) \geq (1 - 7\epsilon)\opt. \qedhere
    \end{equation*}
\end{proof}

\begin{claim} \label{clm:small-dependent-edges}
    Let $\tau^-$ and $\tau^+$ be as in \cref{clm:taus},
    and let $D$ be the set of non-crucial edges $(u, v)$ for which $\delta(u, v) \geq (\epsilon p)^{15}$. Then, it holds
    $$
        q_D \leq \epsilon^2 n \leq \epsilon \opt.
    $$
\end{claim}
\begin{proof}
    For any fixed vertex $u$, we have (by item \ref{item:correlated-sets} of \cref{lem:stochastic-matching-lca})
    $$
    \sum_v \delta(u, v) = \E[\psi(u)] \leq \Delta^{\poly(1/\epsilon)}.
    $$
    Consequently, there are at most $\Delta^{\poly(1/\epsilon)} / (\epsilon p)^{15}$  edges $(u, v)$ with $\delta(u, v) > (\epsilon p)^{15}$.
    Considering this for all $u$, there are at most a total $n \Delta^{\poly(1/\epsilon)} / (\epsilon p)^{15}$ edges $(u, v)$ with $\delta(u, v) > (\epsilon p)^{15}$, i.e.\ 
    $$
    \card{D} \leq n \Delta^{\poly(1/\epsilon)} / (\epsilon p)^{15}.
    $$

    Finally, recall that all the edges in $D$ are non-crucial. That is, $q_e \leq \tau^-$ for all $e \in D$. It follows:
    \begin{equation*}
    q_D = \sum_{e \in D} q_e \leq \tau^- \card{D} \leq \frac{\tau^- \Delta^{\poly(1/\epsilon)}}{(\epsilon p)^{15}} n \leq \epsilon^2 n \leq \epsilon\opt. \qedhere
    \end{equation*}
\end{proof}

\begin{claim} \label{clm:x-tail-bound}
   Let $x$ be as defined in \cref{def:x}. For any vertex $v \in V$, it holds that 
    $$\Pr[x_v \geq 1 + 2\epsilon] \leq \epsilon^2.$$
\end{claim}
\begin{proof} 
We fix the values of $f$, and analyze $\Pr[x_v \geq 1 + 2\epsilon \mid v \in V(M_C)]$ and $\Pr[x_v \geq 1 + 2\epsilon \mid v \notin V(M_C)]$ separately.
When $v \in V(M_C)$, $x$ equals $1$ on the edge of $M_C$ that matches $v$, and is zero on the other edges of $v$.
Therefore $x_v$ is always $1$ in this case, and we have:
$$
\Pr[x_v \geq 1 + 2\epsilon \mid v \in V(M_C)] = 0.
$$

Now, consider the case where $v \notin V(M_C)$.
If $\Pr[v \notin V(M_C)] < \epsilon^2$, then $x_v$ is equal to zero.
Therefore, we assume $Pr[v \notin V(M_C)] \geq \epsilon^2$.
To bound $\Pr[x_v > 1 + 2\epsilon \mid v \notin V(M_C)]$,
we use Chebyshev's inequality.
\begin{equation}
\Pr[x_v > 1+2\epsilon \mid v \notin V(M_C)] 
\leq \Pr[x_v - \mu  > \epsilon \mid v \notin V(M_C)] 
\leq \frac{\var(x_v \mid v \notin V(M_C)}{\epsilon^2},
\label{eq:x-tail-cheybshev}
\end{equation}
where $\mu = \E[x_v \mid v \notin V(M_C)]$ is at most $1 + \epsilon$ as shown in \cref{clm:expectation_ub}.

The variance of $x_v$ can be bounded as follows.
Let $e_1, \ldots e_k$ be the non-crucial edges adjacent to $v$ such that $\Pr[u_i \notin V(M_C)] \geq \epsilon^2$ and $\delta(v,u_i)\leq (\epsilon p)^{15}$, where $u_i$ is the other endpoint of $e_i$.
It holds
$$
x_v = \sum_{1 \leq i \leq k} x_{e_i}.
$$
Hence,
\begin{equation}
\var(x_v \mid v \notin V(M_C))
= \sum_{1 \leq i, j \leq k} \cov(x_{e_i}, x_{e_j} \mid v \notin V(M_C)). \label{eq:sum-cov}
\end{equation}

Considering that $v \notin V(M_C)$,
the value of $e_i$ is determined by $u_i \notin V(M_C)$ and $e_i \in E_p$.
By \cref{lem:distance-from-independence} the three events $[v \notin V(M_C)]$, $[e_i \in E_p, u_i \notin V(M_C)]$, and $[e_j \in E_p, u_j \notin V(M_C)]$ are $(\delta(v, u_i) + \delta(v, u_j) + \delta(u_i, u_j))$-close to being independent.
Also, when $x_{e_i}$ is non-zero, it equals $\frac{f_{e_i}}{p_{e_i} \Pr[u_i \notin V(M_C)]\Pr[v \notin V(M_C)]} \leq \frac{f_{e_i}}{p \epsilon^4}$.
Therefore, we can utilize \cref{clm:cov_ub} to bound the covariances as follows:
$$
\cov(x_{e_i}, x_{e_j} \mid v \notin V(M_C))
\leq \frac{3 (\delta(v, u_i) + \delta(v, u_j) + \delta(u_i, u_j))}{\Pr[v \notin V(M_C)]} \cdot \frac{f_{e_i}f_{e_j}}{p^2\epsilon^8}.
$$
Plugging this back into \eqref{eq:sum-cov}, we get 
\begin{align}
    \var(x_v \mid v \notin V(M_C))
    &\leq \sum_{1 \leq i, j \leq k} \frac{3 (\delta(v, u_i) + \delta(v, u_j) + \delta(u_i, u_j))}{\Pr[v \notin V(M_C)]} \cdot \frac{f_{e_i}f_{e_j}}{p^2\epsilon^8} \notag \\
    &\leq \sum_{1 \leq i, j \leq k} \frac{6(\epsilon p)^{15} + 3\delta(u_i, u_j)}{p^2\epsilon^{10}} \cdot f_{e_i}f_{e_j}\label{eq:low-cov-v} \\
    &= \sum_{1 \leq i \leq k} \frac{f_{e_i}}{p^2\epsilon^{10}} \sum_{1 \leq j \leq k} \left(6(\epsilon p)^{15} + 3\delta(u_i, u_j)\right)  f_{e_j} \notag \\
    &\leq \sum_{1 \leq i \leq k} \frac{f_{e_i}}{p^2\epsilon^{10}} \sum_{1 \leq j \leq k} \left(6(\epsilon p)^{15}f_{e_j} + 3\frac{\delta(u_i, u_j)}{\sqrt{\epsilon R}}\right) \label{eq:var-ub-1} \\
    &\leq \sum_{1 \leq i \leq k} \frac{f_{e_i}}{p^2\epsilon^{10}} \left(6(\epsilon p)^{15} + \frac{(1/\tau^+)^{\poly(1/\epsilon)}}{\sqrt{\epsilon R}}\right) \label{eq:var-ub-2} \\
    &\leq \frac{6(\epsilon p)^{15}}{p^2\epsilon^{10}} + \frac{(1/\tau^+)^{\poly(1/\epsilon)}}{p^2\epsilon^{10} \sqrt{\epsilon R}} \label{eq:var-ub-3} \\
    &\leq \epsilon^4. \label{eq:var-ub-4}
\end{align}
Here, \eqref{eq:low-cov-v} follows from $\delta(v, u_i), \delta(v, u_j) \leq (\epsilon p)^{15}$ and $\Pr[v \notin V(M_C) \geq \epsilon^2]$;
\eqref{eq:var-ub-1} holds since ${f_{e_j} \leq \frac{1}{\sqrt{\epsilon R}}}$ (\cref{lem:f});
\eqref{eq:var-ub-2} results from $\sum_{j}f_{e_j} \leq 1$ (\cref{lem:f}), and
\begin{equation}
    \sum_{1 \leq j\leq k} \delta(u_i, u_j) 
    = \sum_{1 \leq j\leq k} \Pr[Q^+(u_i) \cap Q^+(u_j) \neq \emptyset] 
    \leq \E[\psi(u_i)] \leq (1/\tau^+)^{\poly(1/\epsilon)}; \tag{\cref{lem:stochastic-matching-lca}}
\end{equation}
and \eqref{eq:var-ub-3} follows again from $\sum_j f_{e_{j}} \leq 1$.
The last inequality \eqref{eq:var-ub-4} follows from $R = \frac{1}{2\tau^-} \gg \frac{1}{(\tau^+)^{\poly(1/\epsilon)}}.$

Finally, plugging back the variance into Chebyshev's inequality \eqref{eq:x-tail-cheybshev}, we obtain:
\begin{equation*}
\Pr[x_v \geq 1 + 2\epsilon \mid v \notin V(M_C)]
\leq \frac{\epsilon^4}{\epsilon^2} = \epsilon^2. \qedhere
\end{equation*}
\end{proof}

\begin{lemma}[Rounding Lemma]
    \label{lem:rounding}
    Let $G=(V, E)$ be any graph and let $x: E \to \mathbb{R}^{\geq 0}$ be a randomized assignment of nonnegative values to the edges of $G$ satisfying the following conditions for a parameter $\epsilon \in [0, 1]$, and all vertices $v \in V$:
    $$
    \E[x_v] \leq 1 
    \qquad \textnormal{and} \qquad
    \Pr[x_v > 1 + \epsilon] \leq \alpha.
    $$
    Then given $x$, there exists an assignment $y \coloneq y(x) : E \to [0, 1]$ such that all the following hold: 
    \begin{enumerate}
        \item $y_e \leq x_e$ for all $e \in E$.
        \item $y_v \leq 1$ for all $v \in V$ (in other words, $y$ is a fractional matching of $G$).
        \item $\E[|y|] \geq (1 - \epsilon)\E[|x|] - \alpha n$.
    \end{enumerate}
\end{lemma}

\begin{proof}
    We define $y$ as follows:
    $$
    y_e = \begin{cases} 
        \frac{x_e}{1 + \epsilon}, & \text{if  $x_u, x_v \leq 1 + \epsilon$}, \\
        0, & \text{otherwise}.
        \end{cases}        
    $$
    Then, $y$ satisfies the first two properties in the statement trivially.

    The third property can be proven as follows:
    \begin{align*}
        \E[\card{y}] &\geq \frac{1}{1 + \epsilon}\E\left[\sum_{e = (u, v)} 
        x_e (1 - \bone_{x_u > 1 + \epsilon \text{ or } x_v > 1 + \epsilon}) \right] \\
        &= (1 - \epsilon)\left(\E[\card{x}] - \sum_v \Pr[x_v > 1 + \epsilon]\right) \\
        &\geq (1 - \epsilon) \E[\card{x}] - \alpha n. \qedhere
    \end{align*} 
\end{proof}

For the final step, we construct fractional matching $y$ by applying the Rounding Lemma (\cref{lem:rounding}) to $(1 - \epsilon)x$. Note that $(1-\epsilon)x$ satisfies the condition of the Rounding Lemma by \cref{clm:x-tail-bound}, i.e.\ 
$$
\Pr[(1 - \epsilon)x_v \geq 1 + \epsilon]
\leq \Pr[x_v \geq 1 + 2\epsilon] \leq \epsilon^2.
$$
It remains to show that $\mu(H_p)$ is (almost) as large as $\card{y}$.
We do so by showing that $y$ satisfies the blossom inequalities for small sets.

\begin{proof}[Proof of \cref{thm:main}]
    It holds $\E[\card{y}] \geq (1-\epsilon)\E[\card{(1-\epsilon)x}] - \epsilon^2n   \geq (1 - 10\epsilon)\opt$; The first inequality follows from the Rounding Lemma, and the second from \cref{clm:x-expectation-lb}.
    Therefore, it suffices to show $\E[\mu(H_p)] \geq \E[\card{y}]$.
    To do so, we prove $H_p$ always has a matching of size $(1-\epsilon)\card{y}$ by arguing that $y$ satisfies the blossom inequalities for vertex sets of size at most $\frac{1}{\epsilon}$.
    This implies $\mu(H_p) \geq (1-\epsilon) \card{y}$ by \cref{lem:blossom}.

    Take any vertex set $S \subseteq V$ with $\card{S} = 2k + 1 \leq \frac{1}{\epsilon}$, where $k \geq 0$ is an integer. We show $y(S) \leq k$. Here, $y(S)$ is used to denote sum of $y$ on $E[S]$, the set of edges with both endpoints in $S$:
    $$
    y(S) \coloneq \sum_{e \in E[S]} y_e = \sum_{\substack{(u, v) \in E \\ u, v \in S}} y_{(u, v)}.
    $$
    First, we show the blossom inequality holds when $M_C \cap E[S] = \emptyset$.
    In this case, for every edge $e \in E[S]$ it holds 
    $$
    y_e \leq \frac{1}{\sqrt{\epsilon R}} \leq \epsilon^2.
    $$
    Consequently,
    $$
    y(S) \leq \card{E[S]} \cdot \epsilon^2 \leq \frac{1}{\epsilon^2} \epsilon^2 = 1.
    $$
    This is at most $k$ for $k \geq 1$, and the blossom inequality is trivial for $k = 0$, concluding the case where $M_C \cap E[S] = \emptyset$.
    
    Now consider the case where $S$ includes some edges of $M_C$.
    Let $M$ be $M_C \cap E[S]$, the edges of $M_C$ that are inside $S$.
    Observe that $y_e$ is equal to $1$ on the edges of $M$ and zero on those adjacent to $M$. Hence, we have
    $$
    y(S) = y_M + y( S \setminus V(M_C) ) = \card{M} + y( S \setminus V(M) ).
    $$
    Here, $S \setminus V(M)$ is another vertex set with no edges of $M_C$ inside. The cardinality of $S \setminus V(M)$ is $2k + 1 - 2\card{M} = 2(k - \card{M}) + 1$.
    Therefore, per the analysis of the other case,
    $$
    y(S \setminus V(M)) \leq k - \card{M},
    $$
    which implies
    $$
    y(S) \leq \card{M} + (k - \card{M}) \leq k.
    $$
    As a result, by \cref{lem:blossom}, it holds $\mu(H_p) \geq (1-\epsilon)\card{y}$, and 
    \begin{equation*}
    \E[\mu(H_p)] \geq (1 - \epsilon)\E[\card{y}] \geq (1 - 11\epsilon)\opt.
    \end{equation*}
    The theorem follows from a rescaling of $\epsilon$.
\end{proof}

%% file: stochastic-matching-lca.tex
\section{Proof of \cref{lem:stochastic-matching-lca}}

This section is devoted to the proof of \cref{lem:stochastic-matching-lca}. First, in \cref{subsec:stochastic-matching-generic}, we give a generic characterization of an algorithm $\mB$ that satisfies the first two conditions of \cref{lem:stochastic-matching-lca}.
Then, in section \cref{subsec:stochastic-matching-queries}, we show how it can be implemented as an LCA. This is the crux of our analysis where we show the in-queries and the out-queries are bounded in expectation, thereby proving that the LCA implementation satisfies the last condition of \cref{lem:stochastic-matching-lca} (due to \cref{lem:correlated_ub}).

Our algorithm follows the framework of \cite{BehnezhadDH20}, i.e.\ using augmenting hyperwalks across independent realizations to compute a matching. We include some of their definitions here.

\begin{definition}[Profiles and hyperwalks] 
    Given a graph $G$, a profile $P = ((\mG_0, M_0), \ldots (\mG_\alpha, M_\alpha))$ is a sequence of pairs where each $\mG_i$ is a subgraph of $G$, and each $M_i$ is a matching in $\mG_i$.

    A hyperwalk is a sequence of edge-integer pairs $W = ( (e_1, s_1), \ldots, (e_k, s_k) )$ such that $(e_1, e_2, \ldots, e_k)$ is a walk on $G$, and $s_i \in \{0, 1, \ldots, \alpha\}$ for all $i \in k$. Two hyperwalks are considered adjacent if they share a vertex, i.e.\ if they include edges that are adjacent.

    The result of applying a hyperwalk $W$ to a profile $P$ is denoted by $P \oplus W = ((\mG_0, M'_0), \ldots, (\mG_\alpha, M'_\alpha))$, where
    $$
    M'_i = M_i \cup \{e_j \mid \text{$j$ is odd, and $s_j = i$}\}
    \setminus \{e_j \mid \text{$j$ is even, and $s_j = i$}\}.
    $$
    That is, $M_i'$ is obtained by adding the odd edges $e_j$ of $W$ such that $s_j = i$ to $M_i$, and removing the even edges of $W$ where $s_j = i$.
\end{definition}

We plug in the value $\alpha := \frac{1}{\epsilon^7} - 1$, and prove two useful claims about hyperwalks:
\begin{claim}
    \label{clm:paths-containing-edge}
    \label{clm:paths-containing-vertex}
    For any vertex $v$, there are at most $\Delta^{O((1/\epsilon) \log (1/\epsilon))}$ hyperwalks of length $\frac{2}{\epsilon}$ that contain $v$.
\end{claim}
\begin{proof}
    Fix the position of $v$ on the hyperwalk, there are $\frac{2}{\epsilon} + 1$ possible positions.
    Having fixed the position, there are at most $\Delta^{\frac{2}{\epsilon}}$ ways to choose the edges of the hyperwalk,
    and at most $\alpha^{2/\epsilon}$ ways to choose the integers $\{s_i\}_{i\in\{0, \ldots, \alpha\}}$.
    Therefore, there are at most $\left(\frac{2}{\epsilon} + 1\right) \cdot \Delta^{2/\epsilon} \cdot \alpha^{2/\epsilon} = \Delta^{O((1/\epsilon) \log (1/\epsilon))}$ hyperwalks containing $v$.
\end{proof}

\begin{claim}
    \label{clm:hyperwalk-max-degree}
    A hyperwalk $W$ of length at most $\frac{2}{\epsilon}$ is adjacent to at most $\tDelta \coloneq \Delta^{O((1/\epsilon) \log (1/\epsilon))}$ other hyperwalks of length at most $\frac{2}{\epsilon}$.
\end{claim}
\begin{proof}
    Recall that two hyperwalks are adjacent if they share a vertex.
    There are at most $\frac{2}{\epsilon} + 1$ vertices on $W$, and each of them appears in at most $\Delta^{O((1/\epsilon) \log (1/\epsilon))}$ hyperwalks $W'$ of length at most $\frac{2}{\epsilon}$. Therefore, $W$ is adjacent to at most $\left(\frac{2}{\epsilon} + 1\right) \cdot \Delta^{O((1/\epsilon) \log (1/\epsilon))} = \Delta^{O((1/\epsilon) \log (1/\epsilon))}$ other hyperwalks.
\end{proof}
We also use $\tDelta$ as an upperbound for the number of hyperwalks containing a fixed vertex/edge.

\begin{definition}[Augmenting hyperwalks] 
For a profile $P$ and a vertex $v$, let $d_P(v) = \card{\{i \mid v \in V(M_i) \}}$. A hyperwalk $W$ is augmenting with respect to $P$ if the following hold:
\begin{enumerate}
    \item applying $W$ to $P$ results in a profile, i.e.\ each $M'_i$ in $P \oplus W$ is a matching,
    \item for the first or last vertex $v$ in the walk, it holds $d_{P \oplus W}(v) = d_P(v) + 1$, and
    \item for all the internal nodes $v$, it holds that $d_{P \oplus W}(v) = d_P(v)$.
\end{enumerate}
    
\end{definition}

\input{stochastic-matching-generic}
\input{stochastic-matching-queries}

%% file: stochastic-matching-generic.tex
\subsection{A Generic Description of $\mB$}
\label{subsec:stochastic-matching-generic}

Fixing some graph $G$, the matching $\mB(H, r)$ is computed recursively, where $H$ is a subgraph of $G$ and $r$ is the recursion level.
In this section and \cref{subsec:stochastic-matching-queries}, we argue that $\mB(G_p, 1/\epsilon^9)$ achieves the guarantees of \cref{lem:stochastic-matching-lca}.
We use $\mB(G_p)$ to refer to $\mB(G_p, 1/\epsilon^9)$.

For $r = 0$, $\mB(H, r)$ is empty.
For $r > 0$, first, a profile $P = ((\mG_0, M_0), \ldots, (\mG_\alpha, M_\alpha))$ is computed:
Let $\mG_0 \coloneq H$,
and let $\mG_1, \ldots, \mG_\alpha$ be independent realizations of $G$ with the same sampling probability, where $\alpha = \frac{1}{\epsilon^7} - 1$.
The matchings are computed recursively, i.e.\ $M_i = \mB(\mG_i, r-1)$.

Finally, a set of hyperwalks $I$ is applied to $P$ as follows. 
We call a vertex $v$ \emph{unsaturated} on level $r-1$ if 
$$\Pr[v \in V(\mB(G_p, r-1))] < \Pr[v \in V(\mA(G_p))] - 2\epsilon^2.$$
Here, the probability is taken over the randomness of $G_p$ and the algorithms.
Let $I$ be a $(1-\epsilon)$-approximate MIS of the augmenting hyperwalks of length $\frac{2}{\epsilon}$ that start and end at unsaturated vertices. $I$ is applied to $P$ to obtain $P' = P \Delta I = ((\mG_0, M'_0), \ldots, (\mG_\alpha, M'_\alpha))$. Here, $P \Delta I$ denotes the result of applying the walks $W \in I$ to $P$ one by one. Observe that the order of applying them does not matter since $I$ is an independent set. The output of the algorithm $\mB(H, r)$ is $M'_0$ (see \Cref{alg:generic-alg}).

\begin{algorithm}\caption{The Generic Version of $\mB(H, r)$}\label{alg:generic-alg}
    \textbf{Input:} a subgraph $H \subseteq G$, and recursion level $r$.

    \If{$r = 0$}{\Return $\emptyset$.}

    Let $\alpha \coloneq \frac{1}{\epsilon^7} - 1$, let $\mG_0 \coloneq H$, and let $\mG_1, \ldots, \mG_\alpha$ be independent realizations of $G$.
    
    \For{$i \in \{0, 1, \ldots, \alpha\}$}{
        Let $M_i \coloneq \mB(\mG_i, r-1)$.
    }

    Let profile $P \coloneq ((\mG_0, M_0), \ldots, (\mG_\alpha, M_\alpha))$.

    Let $I$ be any $(1 - \epsilon)$-approximate maximal independent set of the augmenting hyperwalks (w.r.t.\ $P$) that start and end in unsaturated vertices and have length at most $\frac{2}{\epsilon}$. \label{step:approx-mis}

    Let $P' \coloneq P \Delta I = ((\mG_0, M'_0), \ldots, (\mG_\alpha, M'_\alpha))$.

    \Return $M'_0$.
\end{algorithm}

We note that $\mB(G_p)$ satisfies the first two conditions of $\cref{lem:stochastic-matching-lca}$.

\begin{claim}[\cite{BehnezhadDH20}] \label{clm:BDH-expectations}
    The matching $\mB(G_p, r)$ satisfies the following:
    \begin{enumerate}
        \item For any $r \geq 0$, it holds that $\Pr[v \in \mc{B}(G_p, r)] \leq \Pr[v \in \mc{A}(G_p)],$ and
        \item for $r \geq \frac{1}{\epsilon^9}$,
$\E[|\mc{B}(G_p, r)|] \geq \E[|\mc{A}(G_p)|] - \epsilon n$.
    \end{enumerate}
\end{claim}

%% file: stochastic-matching-queries.tex
\subsection{LCA Implementation and Bounding the In/Out-queries}
\label{subsec:stochastic-matching-queries}

In this section, we show how \Cref{alg:generic-alg} can be implemented as an LCA. 
We bound the expected number of in-queries and out-queries for every vertex $v$, thereby bounding the expected size of the correlated sets.

\subsubsection{Background (Maximal Independent Set)}

A key ingredient is the sublinear algorithm for randomized greedy maximal independent set ($\RGMIS$), due to \cite{YoshidaYI12}.
Given a graph $G$, it draws a random permutation $\pi$ over the vertices, and computes the greedy maximal matching using permutation $\pi$, denoted by $\GMIS(G, \pi)$.
The $\GMIS$ is defined by the following greedy algorithm: Start with the empty-set $I$, and go over the vertices in the order they appear in $\pi$. For a vertex $v$, if none of its neighbors appear in $I$, then let $I \gets I \cup \{v\}$.

Remarkably, $\GMIS$ can be implemented as a sublinear algorithm or (with approximation) as an LCA. For a vertex $v$ let $\pi(v) \in [n]$ be its rank in $\pi$. It can be seen that $v \in \GMIS(G, \pi)$ if and only if none of its neighbors $u$ with $\pi(u) < \pi(v)$ appear in $\GMIS(G, \pi)$. 
Therefore, membership in $\GMIS(G, \pi)$ can be tested in the sublinear setting by recursively checking the neighbors of smaller rank.
A crucial observation of \cite{YoshidaYI12}, was that the vertices with smaller ranks are more likely to be in the independent set.
Consequently, going over the neighbors in increasing order of ranks (w.r.t.\ $\pi$) and stopping as soon as one of them is in the independent set, results in a much better running time (see \Cref{alg:rgmis}).

\begin{algorithm}\caption{$\GMIS(G, \pi, v)$: A sublinear algorithm that asserts membership in the greedy maximal independent set, given a graph $G$ and a permutation $\pi$ (\cite{YoshidaYI12})}\label{alg:rgmis}
\For{neighbors $u \in N(v)$ in increasing order of $\pi$} {
    \If{$\GMIS(G, \pi, u)$}{
        \Return false
    }
}
\Return true
\end{algorithm}

\begin{lemma}[\cite{YoshidaYI12}] \label{lem:rgmis}
    Fixing any vertex $v$, take a random permutation $\pi$ and compute $\GMIS(G, \pi, u)$ for all vertices $u \in V(G)$. The expected number of recursive calls made to $\RGMIS(G, \pi, v)$ is $\Delta$, where $O(\Delta)$ is the maximum degree in $G$.
\end{lemma}

A direct implication of \cref{lem:rgmis} is that for \emph{a random vertex} $v$, $\RGMIS(G, v)$ makes $O(\Delta)$ recursive calls, and hence runs in time $O(\Delta^2)$.
We use a modified version of $\RGMIS$ that halts as soon as $O(\Delta^2 / \epsilon)$ recursive calls are made, implemented as an LCA, to get a $(1 - \epsilon)$-approximate MIS.

\begin{definition}
    The truncated randomized greedy maximal independent set ($\TMIS$), is an LCA that simulates $\RGMIS$. The random tape of each vertex contains its rank (e.g.\ a number in $[0, 1]$). If $\RGMIS$ halts within $O(\Delta^2 /\epsilon)$ recursive calls, $\TMIS$ returns the same output.
    Otherwise, $\TMIS$ halts as soon as the threshold is reached and returns false, i.e.\ reports that the vertex is not in the independent set.
\end{definition}

\begin{claim} \label{clm:tmis}
    $\TMIS$ has the following properties:
    \begin{enumerate}
        \item $\TMIS$  returns a $(1 - \epsilon)$-approximate MIS in expectation.
        \item For any vertex, the number of recursive calls is $O(\Delta^2 / \epsilon)$ deterministically, i.e.\ the number of out-queries is $O(\Delta^3/\epsilon)$.
        \item For any vertex, the expected number of in-queries is $O(\Delta)$, where the expectation is taken over the random tapes. \label{item:tmis-expected-in-query}
    \end{enumerate}
\end{claim}
\begin{proof}
Without truncation, the algorithm would have returned a maximal independent set. Let $I$ be that independent set.
Note that $\card{I} \geq \frac{n}{\Delta + 1}$ since for every vertex in the independent set, there are at most $\Delta$ vertices not in it.
To prove the first property it suffices to show that in expectation, the algorithm truncates for at most $\epsilon \card{I}$.
By \cref{lem:rgmis}, the total number of recursive calls when $\RGMIS$ is queried for every vertex is at most $O(n\Delta)$ in expectation.
As a result, in expectation, there are at most $O(\epsilon n / \Delta) \leq \epsilon\card{I}$ vertices that make more than $O(\Delta^2 / \epsilon)$ queries.

The second property holds by the definition of $\TMIS$,
and the third property follows from the same guarantee for $\RGMIS$ (\cref{lem:rgmis}).
\end{proof}

\subsubsection{Implementation}
First, we describe the model.
The LCA receives inquiries about whether an edge $e$ is in the matching returned by $\mB(G_p, r)$.
For technical reasons, which will become clear shortly, the whole graph $G$ is known to the LCA.
However, the random tapes associated with each edge are known only locally, and the algorithm has to query an edge to discover its random tape.
The random tapes determine whether each edge is realized in $G_p$,
and also serve as the source of randomness for $\mB$.

 Observe that everything could be described in the vertex-oriented version of LCAs. For example, the inquiries could be regarding whether a vertex $v$ is matched in $\mB(G_p, r)$ and, if so, to which neighbor. Also, the fact that an edge is realized could be determined by the random tape of one of its endpoints (e.g.\ the one with the smaller ID). Nonetheless, we take the edge-oriented approach here as it is the more natural setting for the matching problem.

 Let us now consider the recursive structure of the algorithm.
 To calculate the final matching $M_0'$ on $G_p$, the algorithm first computes $\alpha + 1$ matchings $M_0, \ldots, M_\alpha$ using $\mB(\cdot, r - 1)$.
 Here, $M_0$ is a matching in $G_p$, the input realization,
 whereas $M_1, \ldots, M_\alpha$ are each in an independent realization $\mG_i$ of $G$.
 The former is the randomization of the input,
 and the latter is the randomization of the algorithm.
 Both are stored locally in the random tapes of the edges.
 Each of these matchings $M_i$ in turn creates $\alpha + 1$ matching $M^{i}_0, \ldots, M^{i}_\alpha$ using $\mB(\cdot, r - 2)$,
 and each of them creates $\alpha + 1$ matchings, and so on.

 To implement this recursion as an LCA, we use three functions:
 $\isInMatching(H, e, r)$ outputs whether $e$ is in the matching $\mB(H, r)$,
 $\isInMIS(H, W, r)$ uses $\TMIS$ to assert if a hyperwalk $W$ is in the independent set $I$ (in step \ref{step:approx-mis}) when computing $\mB(H, r)$,
 and $\isValid(H, W, r)$ determines whether a hyperwalk $W$ is an augmenting hyperwalk w.r.t.\ the profile $P = ( (\mG_0, M_0), \ldots, (\mG_\alpha, M_\alpha) )$ when computing $\mB(H, r)$ and if its endpoints are saturated.
 We go over each of them in detail (see also \Cref{alg:is-in-matching,alg:is-in-MIS,alg:is-valid}).

 The main subroutine, $\isInMatching(H, e, r)$, first checks if $e \in M_0$ using $\isInMatching(H, e, r-1)$. Then, it goes over all the hyperwalks $W$ of length at most $\frac{2}{\epsilon}$ that include $e$ and if $W$ is in $I$ (which is checked using $\isInMIS(H, W, r)$), it applies $W$ to $e$. That is, $e \in M'_0$ can be determined by checking $e \in M_0$ and whether applying $I$ adds or removes $e$.

 \begin{algorithm}
     \caption{$\isInMatching(H, e, r)$: returns true if $e \in \mB(H, r)$}
     \label{alg:is-in-matching}

    output $\gets \isInMatching(H, e, r - 1)$ 
    
    \For{every hyperwalk $W$ of length at most $2/\epsilon$ containing $e$} {
        \If{$\isInMIS(H, W, r)$}{
            apply $W$ to adjust output
        }
    }

    \Return ans
 \end{algorithm}

 To determine if a hyperwalk $W$ is in $I$ when computing $\mB(H, r)$,
 $\isInMIS(H, W, r)$ runs $\TMIS$ on the agumenting hyperwalks.
 Ideally, we would have run $\TMIS$ on the graph where each vertex is an augmenting hyperwalk, and two vertices are adjacent if the corresponding hyperwalks share a vertex.
 However, whether a hyperwalk is augmenting depends on the profile $P = ((\mG_0, M_0), \ldots, (\mG_\alpha, M_\alpha))$, which is not readily available to the algorithm.
 Therefore, a slight variation of the algorithm is used. That is, $\TMIS$ runs on the graph where the vertex set is the set of \emph{all} hyperwalks of length at most $\frac{2}{\epsilon}$,
 and the algorithm checks if a hyperwalk is augmenting w.r.t.\ $P$ on the fly. Note that the maximum degree in the graph of hyperwalks is at most $\tDelta = \Delta^{O((1/\epsilon) \log (1/\epsilon))}$ (\cref{clm:hyperwalk-max-degree}), therefore $\TMIS$ should terminate if more than $\tDelta^2/\epsilon$ recursive calls are made to valid (i.e.\ augmenting) hyperwalks. This involves exploring the rank and validity of at most $\tDelta^3 / \epsilon$ hyperwalks.

 \begin{algorithm}
    \caption{$\isInMIS(H, W, r)$: returns true if $W \in I$ when computing $\mB(H, r)$}
    \label{alg:is-in-MIS}
    \If{the recursive calls to $\isInMIS(H, \cdot, r)$ made by descendants of this call exceeds $\tDelta^3/\epsilon$}{
        \Return false
    }

    \If{not $\isValid(H, W, r)$}{
        \Return false
    }

    \For{hyperwalks $W'$ of length $\leq 2/\epsilon$ adjacent to $W$ in increasing order of rank} {
        \If{$\isInMIS(H, W', r)$}{
            \Return false
        }
    }
    
    \Return true
 \end{algorithm}

 Finally, $\isValid(H, W, r)$ first checks that the endpoints of $W$ are unsaturated (on level $r-1$),
 where recall a vertex $v$ is unsaturated if
 $$\Pr[v \in V(\mB(G_p, r-1))] < \Pr[v \in V(\mA(G_p))] - 2\epsilon^2.$$
 Note that both these values are determined by the structure of $G$ which is completely known to the algorithm.

 $\isValid(H, W, r)$ also verifies that $W$ is a valid augmenting path, i.e.\ $(1)$ the subgraphs $\{M_i''\}_i$ in $P \oplus W = ((\mG_0, M_0''), \ldots (\mG_\alpha, M_\alpha''))$ are matchings,
 $(2)$ for each internal node $v$, it holds $d_{P\oplus W}(v) = d_P(v)$, and $(3)$ for the two endpoints, $d_{P\oplus W}(v) = d_P(v) + 1$. Observe, that applying $W$ to $P$ only changes the status of the edges in $W$ (in terms of which matching they appear in), and $d_P(v)$ can only change for the vertices in $W$.
  Therefore, to check if the three conditions hold, it suffices to compute $\isInMatching(\mG_i, e, r-1)$ for all $i \in \{0, 1, \ldots \alpha\}$ and every $e$ that is in $W$ or adjacent to it.
 
 \begin{algorithm}
     \caption{$\isValid(H, W, r)$: returns true if hyperwalk $W$ is augmenting w.r.t.\ $P$ when computing $\mB(h, r)$ and its endpoints are unsturated}
     \label{alg:is-valid}
     \If{either endpoint of $W$ is saturated}{
        \Return false
     }

     \For(\tcp*[f]{compute $P$ locally around $W$}){vertices $u$ in $W$}{ 
        \For{edges $e$ adjacent to $u$}{
            \For{$i \in \{0, 1, \ldots, \alpha\}$}{
                compute $\isInMatching(\mG_i, e, r - 1)$
            }
        }
     }

     Use the local information about $P$ to assert if $W$ is an augmenting hyperwalk, and return the result
 \end{algorithm}

\subsubsection{In-queries, Out-queries, and Correlated Sets}

\begin{claim}
    \label{clm:low-out-query}
    Consider the LCA for $\mB(G_p)$ (\Cref{alg:is-in-matching} with $r = 1/\epsilon^9$).
    For every edge $e$, it holds deterministically 
    $$q^+(e) \leq \left(\frac{\tDelta^4}{\epsilon} \cdot \frac{2\Delta}{\epsilon^{10}}\right)^{1/\epsilon^9} = \Delta^{\poly(1/\epsilon)}.$$
\end{claim}

\begin{proof}
    Observe that $Q^+(e)$ is the set of edges $e'$ for which a call to $\isInMatching(\cdot, e', 0)$ is made.
    Let us examine the number of recursive calls made to $\isInMatching(\cdot, \cdot, j-1)$ as a result of computing $\isInMatching(H, e, j)$, for $1 \leq j \leq r$.
    Recall that the edge $e$ is contained in at most $\tDelta = \Delta^{O((1/\epsilon)\log(1/\epsilon))}$ hyperwalks (\cref{clm:paths-containing-edge}).
    Therefore, there are at most $\tDelta$ direct calls made to $\isInMIS(H, \cdot, j)$ from $\isInMatching(H, e, j)$.
    Each of these calls terminates before making $\tDelta^3/\epsilon$ recursive calls to $\isInMIS(H, \cdot, j)$.
    Finally, each of the calls to $\isInMIS(H, \cdot, j)$ (direct or recursive) makes a call to $\isValid(H, \cdot, j)$ which in turn makes at most $\frac{2\Delta}{\epsilon^{10}}$ calls to $\isInMatching(\cdot, \cdot, j-1)$, i.e.\ $\alpha\Delta$ calls for each of the $\frac{2}{\epsilon}$ vertices in the hyperwalk.

    Hence, the total number of calls made to $\isInMatching(\cdot, \cdot, j-1)$ as a result of computing $\isInMatching(H, e, j)$ is at most 
    $$ 
    \frac{\tDelta^4}{\epsilon} \cdot \frac{2\Delta}{\epsilon^{10}}.
    $$
    It follows by induction that for all $1 \leq j \leq r$, the number of calls made to $\isInMatching(\cdot, \cdot, r - j)$ as a result of computing $\isInMatching(H, e, r)$ is at most $\left(\frac{\tDelta^4}{\epsilon} \cdot \frac{2\Delta}{\epsilon^{10}}\right)^j$. Therefore, computing $\isInMatching(G_p, e, 1/\epsilon^9)$ (which is the LCA for $\mB(G_p)$) takes at most $\left(\frac{\tDelta^4}{\epsilon} \cdot \frac{2\Delta}{\epsilon^{10}}\right)^{1/\epsilon^9}$ queries to $\isInMatching(\cdot, \cdot, 0)$, which is an upperbound for the out-queries.
\end{proof}

\begin{claim}
    \label{clm:low-in-query}
    Consider the LCA for $\mB(G_p)$ (\Cref{alg:is-in-matching} with $r = 1/\epsilon^9$).
    For every edge $e$, it holds 
    $$\E_{\mB}[q^-(e)] \leq \Delta^{\poly(1/\epsilon)},$$
    where the expectation is taken over the randomness of the algorithm, and the bound holds for any fixed realization $G_p$.
\end{claim}
\begin{proof}
    The randomness of the algorithm can be divided into two parts:
    the independently drawn realizations,
    and the ranks of the hyperwalks for $\isInMIS$.
    We condition on the realizations.
    That is, for the remainder of the proof we fix all the realizations drawn by the algorithm, and show that the bound holds even when the expectation is taken only over the ranks.
    Note that the number of in-queries $q^-(e)$ can be bounded above by the number of calls made to $\isInMatching(\cdot, e, 0)$ when $\isInMatching(G_p, e', r)$ is computed for every $e'$, where $r = 1/\epsilon^9$. We analyze this number.

    When computing $\isInMatching(G_p, \cdot, r)$, the calls to $\isInMIS(G_p, \cdot, r)$ utilize ranks on hyperwalks of $G$.
    Let $\pi_r$ denote these ranks.
    Note that ranks are defined even for invalid (non-augmenting) hyperwalks.
    On the next level of recursion, there are $\alpha + 1$ independent realizations. $\isInMIS$ runs on them independently, i.e.\ for $i \in \{0, 1, \ldots, \alpha\}$ it has a separate set of ranks which is used to eventually compute $M_i$. We use $\pi_{r-1}$ to denote \emph{all} these ranks.
    Similarly for $1 \leq j \leq r$, we use $\pi_{j}$ to denote the ranks of the hyperwalks for the $(\alpha+1)^{r-j}$ realizations of the corresponding level (there is no $\pi_0$ since on the bottom-most level $\isInMatching(\cdot, \cdot, 0)$ returns false without computing any augmenting hyperwalks), and $\pi_{\leq j}$ to denote the sequence $(\pi_1, \pi_2, \ldots, \pi_j)$.

    To bound $\E_\mB[q^-(e)]$, we use a counting argument while going over all possible random tapes (i.e.\ sets of ranks) as follows. 
    Let $\Pi_j$ (resp.\ $\Pi_{\leq j}$) be the set of possible values for $\pi_j$ (resp.\ $\pi_{\leq j}$).
    For $0 \leq j \leq r$ and $\pi_{\leq j} \in \Pi_{\leq j}$,
    consider going over all the edges $e'$ and the sets of ranks $\pi \in \Pi_{\leq r}$ such that the first $j$ entries are equal to $\pi_{\leq j}$, and computing $\isInMatching(G_p, e, r)$ with the given set of set of ranks $\pi$.
    Let $\qq_j(e, \pi_{\leq j})$ be the number of calls made to $\isInMatching(\cdot, e, j)$ as a result of this process. For brevity, we refer to this as querying $e$ with $\pi_{\leq j}$. For $j = 0$, $\pi_{\leq j}$ is simply a placeholder as there are no ranks for the bottom-most level, and $\qq_0(e, \pi_{\leq 0})$ is the number of queries made to $\isInMIS(\cdot, e, 0)$ while going over \emph{all} the sets of ranks in $\Pi_{\leq r}$.
    That is, $\E_\mB[q^-(e)] = \qq_0(e, \pi_{\leq 0}) / \card{\Pi_{\leq r}}$.

    We show by induction that for $0 \leq j \leq r$, it holds $$\qq_{j}(e, \pi_{\leq j}) \leq \left((\alpha + 1) \frac{2 \tDelta^3}{\epsilon}\right)^{r-j} \card{\Pi_{\leq r}},$$ for all edges $e$ and sets of ranks $\pi_{\leq j} \in \Pi_{\leq j}$.
    This holds by definition for $j = r$,
    i.e.\ for all $e$ and $\pi_{\leq r} \in \Pi_{\leq r}$, we have $\qq(e, \pi_{\leq r}) = 1$.
    Assuming the claim for any $1 \leq j \leq r$, we show it for $j - 1$.

    Fix any $\pi_{\leq j - 1} \in \Pi_{\leq j-1}$, iterate over all $\pi_j \in \Pi_j$, and add it to $\pi_{\leq j - 1}$. This results in some $\pi_{\leq j} \in \Pi_{\leq j}$ which shares the first $j-1$ entries with $\pi_{\leq j-1}$.
    Note that fixing $\pi_{\leq j - 1}$ fully defines the matchings of the bottom $j-1$ levels of recursion,
    in turn determining the validity of the hyperwalks on level $j$.
    Let $t \coloneq \left((\alpha + 1) \frac{2 \tDelta^3}{\epsilon}\right)^{r-j}$.
    By the induction hypothesis, for any edge $e$, it holds $\qq_j(e, \pi_{\leq j}) \leq t$. For simplicity, we assume $\isInMatching(\cdot, e, j)$ has been queried exactly $t$ times with $\pi_{\leq j}$. This can only increase the number queries made to the lower levels.
    
    As a result of the calls to $\isInMatching(\cdot, \cdot, j)$ with $\pi_{\leq j}$, for any hyperwalk $W$, there will be at most $\frac{2}{\epsilon} \cdot t$ direct calls to $\isInMIS(\cdot, W, j)$ with $\pi_{\leq j}$, i.e.\ $t$ calls from each of the $\frac{2}{\epsilon}$ edges in $W$.
    Now we bound the number of (direct or recursive) calls made to $\isInMIS(\cdot, W, j)$ with $\pi_{\leq j - 1}$, by iterating over all possible values of $\pi_j$.
    Recall that $\pi_j$ is the set of ranks used for emulating $\TMIS$ on the valid hyperwalks.
    Due to item \ref{item:tmis-expected-in-query} of \cref{clm:tmis},
    the expected number of queries made to each valid hyperwalk $W$ by $\TMIS$ is at most $\tDelta = \Delta^{O((1/\epsilon)\log(1/\epsilon))}$.
    Taking into account that $\isInMIS$ also has to go over the neighbors of valid hypewalks, we can conclude that the number of queries made to $\isInMIS(\cdot, W, j)$ with $\pi_{\leq j-1}$ is at most 
    $$\tDelta^2\card{\Pi_r} \cdot \frac{2}{\epsilon} \cdot t.$$

    Finally, each of these calls to $\isInMIS(\cdot, W, j)$, recursively calls $\isValid(\cdot, W, j)$, which in turn makes $(\alpha + 1)$ calls to $\isInMatching(\cdot, e, j-1)$ for each $e \in W$.
    Since each edge is contained in at most $\tDelta$ hyperwalks (\cref{clm:paths-containing-edge}),
    the total number of calls made to $\isInMatching(\cdot, e, j-1)$ with $\pi_{\leq j-1}$ (i.e.\ $\qq_{j-1}(e, \pi_{\leq j-1})$) is at most 
    \begin{align*}
    (\alpha + 1)\tDelta \cdot \tDelta^2\card{\Pi_r} \cdot \frac{2}{\epsilon} \cdot t
    &= \left((\alpha + 1) \frac{2 \tDelta^3}{\epsilon}\right) \card{\Pi_r} \card{\Pi_{>r}} \left((\alpha + 1) \frac{2 \tDelta^3}{\epsilon}\right)^{r-j}  \\
    &= \left((\alpha + 1) \frac{2 \tDelta^3}{\epsilon}\right)^{r-(j-1)} \card{\Pi_{>{r-1}}},
    \end{align*}
    which concludes the proof of the induction.

    Putting things together, for any edge $e$, we have
    \begin{align*}
        \E_{\mB}[q^-(e)] &= \qq_0(e, \pi_{\leq 0}) / \card{\Pi_{\leq r}} \\
        &\leq \left((\alpha + 1) \frac{2 \tDelta^3}{\epsilon}\right)^r \\
        &= \Delta^{\poly(1/\epsilon)}. \qedhere
    \end{align*}
\end{proof}

\begin{proof}[Proof of \cref{lem:stochastic-matching-lca}]
    The LCA $\mB$ is described in \Cref{alg:is-in-matching,alg:is-in-MIS,alg:is-valid}.
    Properties \ref{item:vertex-expectation} and \ref{item:matching-expectation} are due to \cref{clm:BDH-expectations}.
    Property \ref{item:correlated-sets} regarding correlated sets follows from combining the bound on the in-queries (\cref{clm:low-in-query}) and the out-queries (\cref{clm:low-out-query}) using $\cref{lem:correlated_ub}$:
    \begin{equation*}
    \E[\psi(v)] \leq Q^+ \cdot Q^- \leq \Delta^{\poly(1/\epsilon)}. \qedhere
    \end{equation*}
\end{proof}

%% file: proof-of-f.tex
\section{Proof of \cref{lem:f}}
\label{sec:ref}
We restate \cref{lem:f}
\lemf*
\begin{proof}
First, we make some definitions. For every edge $e$, let $t_e$ be the fraction of matchings $\MM(\mG_1), \ldots, \MM(\mG_R)$ that contain $e$. That is,
$$
t_e = \frac{\card{\{1 \leq i \leq R \mid e \in \MM(\mG_i)\}}}{R}.
$$
Then, construct $f'_e$ by capping $t_e$ at $\frac{1}{\sqrt{\epsilon R}}$, and zeroing out the crucial edges.
That is,
$$
f'_e = \begin{cases}
    t_e & \text{if $e \in N$ and $t_e \leq \frac{1}{\sqrt{\epsilon R}}$, and} \\
    0 & \text{otherwise.}
\end{cases}
$$
Finally, construct $f$ by scaling down $f'$ and zeroing out all edges for which one of the endpoints $v$ has $f'_v$  larger than $q^N_v$, i.e.,
$$
f_{(u, v)} = \begin{cases}
    (1 - \epsilon) f'_{(u, v)} & \text{if $(1-\epsilon)f'_u \leq q^N_u$ and $(1-\epsilon)f'_v \leq q^N_v$, and} \\
    0 & \text{otherwise.}
\end{cases}
$$

\paragraph{Property 1:}
Recall that by definition, $t_e$ is the fraction of matchings $\MM(\mG_i)$ that include $e$. Therefore, for $e \notin H$, it holds $t_e = 0$. This implies $f_e = 0$ since $f_e \leq f'_e \leq t_e$.

\paragraph{Property 2:}
This property follows from the fact that for any vertex $v$, if $(1 - \epsilon) f'_v > q^N_v$, we let $f_e = 0$ for all edges adjacent to $v$.
Otherwise, we have $f_v = (1 - \epsilon)f'_v \leq q^N_v$.

\paragraph{Property 3:}
It holds $f'_e \leq \frac{1}{\sqrt{\epsilon R}}$ by definition. Hence, the property follows from $f_e \leq f'_e$.

\paragraph{Property 4:}
We have:
\begin{align}
    \E[t_e] &= \frac{1}{R} \E[\card{\{1 \leq i \leq R \mid e \in \MM(\mG_i)\}}] \notag \\
    &= \frac{1}{R} R \E[\bone_{e \in \MM(\mG_1)}] \label{eq:lem-f-linearity} \\
    &= q_e, \notag
\end{align}
where \eqref{eq:lem-f-linearity} follows from linearity of expectation and that $\MM(\mG_i)$ contains $e$ with probability $q_e$ for every $i$.
$\E[f_e] \leq q_e$ follows from $f_e \leq f'_e \leq t_e$.

\paragraph{Property 5:}
First, we give a lower bound for $\E[f'_e]$ for all edges $e \in N$.
It holds by definition:
\begin{equation}
\E[f'_e] = \E[t_e] - \Pr\left[t_e > \frac{1}{\sqrt{\epsilon R}}\right]
\E\left[t_e \; \Big\lvert \; t_e > \frac{1}{\sqrt{\epsilon R}}\right].
\label{eq:property5-fp}
\end{equation}
Since $t_e$ is always at most $1$, this implies
$$
\E[f'_e] \geq \E[t_e] - \Pr\left[t_e > \frac{1}{\sqrt{\epsilon R}}\right] = q_e - \Pr\left[t_e > \frac{1}{\sqrt{\epsilon R}}\right].
$$

To bound $\Pr\left[t_e > \frac{1}{\sqrt{\epsilon R}}\right]$. We use Chebyshev's inequality. Observe that $R t_e$ is a binomial variable with $R$ trials, and success probability $q_e$.
Therefore, 
$$
    \var(t_e) = \frac{1}{R^2} \var(\text{Binom}(R, q_e)) 
    = \frac{1}{R^2} \left(R  q_e  \left(1 - q_e\right)\right) \leq\frac{q_e}{R}.
$$
Consequently,
\begin{align}
\Pr\left[t_e \geq \frac{1}{\sqrt{\epsilon R}}\right]
&\leq \Pr\left[t_e \geq \frac{1}{\sqrt{2\epsilon R}} + \frac{1}{R} \right] \label{eq:property5-algebra}\\
&\leq \Pr\left[t_e - \E[t_e] \geq \frac{1}{\sqrt{2\epsilon R}} \right] \label{eq:property5-expectation-ub}\\
&\leq \frac{\var(t_e)}{1/2 \epsilon R} 
\notag\\
&\leq 2 \epsilon q_e, \notag
\end{align}
where \eqref{eq:property5-algebra} holds because $\frac{1}{\sqrt{2\epsilon R}} + \frac{1}{R} \leq \frac{1}{\sqrt{\epsilon R}}$, and \eqref{eq:property5-expectation-ub} follows from $q_e \leq \tau^- \leq \frac{1}{R}$ (recall $e \in N$).
Plugging this back into \eqref{eq:property5-fp}, one can obtain:
$$
\E[f'_e] \geq (1-2\epsilon)q_e.
$$

Finally, we bound $\E[\card{f}]$ for any vertex $v \in V$, as follows:
\begin{align}
\E[\card{f}]
&\geq (1-\epsilon)\E\left[ \sum_{e} f'_{e} - \sum_u f'_u \cdot \bone_{(1-\epsilon)f'_u > q^N_u}\right] \label{eq:property5-f-1}\\
& \geq (1-\epsilon)\left( (1-2\epsilon)q_N - \sum_u \Pr\left[(1-\epsilon)f'_u > q^N_u\right]\right). \label{eq:property5-f-2}
\end{align}
Here, \eqref{eq:property5-f-1} is true since for every edge $(u, v)$, $f'_{(u, v)}$ appears once in the first sum, and if $(1 - \epsilon)f'_u > q_u^N$ or $(1 - \epsilon)f'_v > q_u^N$, then it appears at least once in the second sum (hence, accounting for zeroing out $f_{(u, v)}$).
\eqref{eq:property5-f-2} follows from $\E[f'_e] \geq (1-2\epsilon)q_e$ and $f'_u \leq 1$.

To bound $\Pr[(1-\epsilon)f'_u \geq q_u^N]$, we first note
\begin{equation}
\Pr\left[(1-\epsilon)f'_u \geq q_u^N\right]
\leq \Pr\left[(1 - \epsilon) t^N_u \geq  q_u^N\right]
\leq \Pr\left[t^N_u \geq q_u^N + \epsilon\right].
\label{eq:property5-probs}
\end{equation}
Observe that $t^N_u$ is the fraction of matchings $\MM(\mG_i)$ that match $u$ using a non-crucial edge.
For each $i$, this happens with probability $q^N_u$.
Therefore, similar to the analysis for $t_e$, it follows:
$$
\E\left[t^N_u\right] = q^N_u \qquad \text{and} \qquad \var\left(t^N_u\right) \leq \frac{q^N_u}{R}.
$$
Applying Cheybshev's inequality, we get
$$
\Pr\left[t^N_u \geq q_u^N + \epsilon\right]
\leq \frac{q^N_u}{R \epsilon^2} \ll \epsilon^2.
$$
Plugging this back into \eqref{eq:property5-probs}, and then into \eqref{eq:property5-f-2}, we can conclude:
\begin{equation*}
\E[\card{f}] \geq (1 - 3\epsilon) q_N - \epsilon^2 n \geq q_N - 4\epsilon \opt,
\end{equation*}
and Property 5 follows from a rescaling of $\epsilon$.
\end{proof}